%% file: lipics-mfcs.tex
\documentclass[a4paper,UKenglish,cleveref, autoref, thm-restate]{lipics-v2021}

\pdfoutput=1
\hideLIPIcs

\input{macros.tex}
\title{Order-invariant cluster first-order logic on graph classes of bounded degree}

\author{Fatemeh Ghasemi}{Univ Paris Est Creteil, LACL, F-94010 Creteil, France }{fatemeh.ghasemi@lacl.fr}{https://orcid.org/0009-0009-6391-1032}{}
\author{Julien Grange}{Univ Paris Est Creteil, LACL, F-94010 Creteil, France}{julien.grange@lacl.fr}{https://orcid.org/0009-0005-0470-1781}{}

\authorrunning{F.~Ghasemi, J.~Grange}

\Copyright{Fatemeh Ghasemi, Julien Grange}



\nolinenumbers

\EventEditors{John Q. Open and Joan R. Access}
\EventNoEds{2}
\EventLongTitle{42nd Conference on Very Important Topics (CVIT 2016)}
\EventShortTitle{CVIT 2016}
\EventAcronym{CVIT}
\EventYear{2016}
\EventDate{December 24--27, 2016}
\EventLocation{Little Whinging, United Kingdom}
\EventLogo{}
\SeriesVolume{42}
\ArticleNo{23}

\begin{document}

\maketitle

\begin{abstract}
We introduce a new logic, called \emph{cluster first-order logic}, a restricted fragment of first-order logic specifically designed to study order invariance. An order-invariant formula is one on a vocabulary that contains an order; however, whether a structure satisfies it or not is independent of the interpretation of the order.

We show that while order-invariant cluster first-order logic can define properties outside the scope of plain first-order logic in general, its expressive power is included in that of first-order logic when it comes to classes of bounded degree.

We establish this result by explicitly constructing linear orders such that similar structures remain similar when they are expanded with these orders. This similarity-preserving, local-to-global approach is technically involved and somewhat counterintuitive, since adding an order usually reveals distinctions that are otherwise hidden due to the locality of first-order logic.
We believe that this work can be a stepping stone toward applying such techniques to plain first-order logic and toward settling the question of the expressive power of order-invariant plain first-order logic.
\end{abstract}

\input{introduction}
\input{Preliminaries}
\input{BoundedDegree}
\input{ModelChecking}

\input{OtherProperties}

\section{Conclusion}

There remain many open problems and directions left for exploration. The gap between order-invariant cluster first-order logic and order-invariant first-order logic should be studied. One can also examine if the presented order can be tweaked and somehow lifted to plain $\FO$ instead of only $\CFO$. The expressive power of order-invariant cluster $\FO$ on more general classes is also of interest; we conjecture that the same upper bound can be shown for \emph{rooted trees} as well. On trees, it is already known that order-invariant \FO collapses to plain \FO with respect to expressive power. However, in \cite{DBLP:journals/jsyml/BenediktS09}, the authors indirectly use the second method and compute orders maintaining similarity of already \FO-similar trees, but for a chain of unbounded length of intermediate structures between the two structures they began with. Their approach highly relies on the good properties of trees and does not seem to be generalisable to more general classes of structures. Nevertheless, we believe that this result could be recovered through a more direct approach, namely by constructing suitable orders on \FO-similar trees so that their equivalence is preserved, and that establishing this result first for \oicfo would be a stepping stone toward this goal. As stated above, we consider the development of such a method to be of independent interest, as it may be extendable to broader graph classes. We also conjecture that \FO is contained in \oicfo in general.

\bibliography{ref.bib}

\appendix

\end{document}

%% file: macros.tex
\usepackage{algorithm}
\usepackage{algorithmic}
\usepackage{relsize,xspace}
\usepackage[bibliography=common]{apxproof}
 \usepackage{xcolor}
 \usepackage{mathtools}
 \usepackage{comment}
\usepackage{microtype}
\usepackage{amsmath}
\usepackage{amssymb}
\usepackage{amsfonts}
\usepackage{stmaryrd}
\usepackage{bm}
\usepackage{tikz}
\usepackage{refcount}
\usepackage{wrapfig}
\usepackage{thmtools}
\usepackage{thm-restate}
\usepackage{xparse}
\usepackage{hyperref}
\usepackage{dsfont}
\usepackage{float}     
\usepackage{tabularx}  
\usepackage{array}     
\usepackage{booktabs}  

\usepackage{marginnote}

\definecolor{sz}{rgb}{0.1,0.2,0.6}
\definecolor{blue}{rgb}{0.1,0.2,0.5}
\definecolor{brown}{rgb}{0.6,0.6,0.2}

\newtheoremrep{theorem}{Theorem}[section]
\newtheoremrep{lemma}{Lemma}[section]
\newtheoremrep{definition}{Definition}[section]
\newtheoremrep{example}{Example}[section]
\newtheoremrep{claim}{Claim}[section]

\newtheoremrep{observation}{Observation}[section]
\newtheoremrep{corollary}{Corollary}[section]

\usepackage{cleveref}

\crefformat{page}{#2page~#1#3}%
\Crefformat{page}{#2Page~#1#3}%
\crefformat{equation}{#2(#1)#3}%
\Crefformat{equation}{#2(#1)#3}%
\crefformat{figure}{#2Figure~#1#3}%
\Crefformat{figure}{#2Figure~#1#3}%
\crefformat{section}{#2Section~#1#3}
\Crefformat{section}{#2Section~#1#3}
\crefformat{chapter}{#2Chapter~#1#3}
\Crefformat{chapter}{#2Chapter~#1#3}
\crefformat{chapter*}{#2Chapter~#1#3}
\Crefformat{chapter*}{#2Chapter~#1#3}
\crefformat{part}{#2Part~#1#3}
\Crefformat{part}{#2Part~#1#3}
\crefformat{enumi}{#2(#1)#3}
\Crefformat{enumi}{#2(#1)#3}

\usepackage{enumerate}

\usepackage{latexsym}

\newcommand{\CCC}{\mathcal{C}}

\newcommand{\nc}[0]{\mathsf{NC}}

\newcommand{\ntype}[3]{\mathcal{N}^{#1}_{#2}(#3)}
\newcommand{\neightype}[0]{\mathsf{NEIGHTYPE}}
\newcommand{\freq}[0]{\mathsf{FREQ}}
\newcommand{{\cont}}[0]{\mathsf{CONT}}
\newcommand{\allcontext}[0]{\mathsf{CONTEXT}}

\newcommand{\bdist}[0]{\mathsf{Bdist}}

\newcommand{\ie}{i.e.\@ }

\newcommand{\N}{\mathbb{N}}

\renewcommand{\epsilon}{\varepsilon}

\newcommand{\FO}{\ensuremath{\mathsf{FO}}\xspace}
\newcommand{\MSO}{\ensuremath{\mathsf{MSO}}\xspace}
\newcommand{\CFO}{\ensuremath{\mathsf{CFO}}\xspace}

\newcommand{\dist}{\mathsf{dist}}

\renewcommand{\mid}{~:~}

\newcommand{\cfo}[0]{\ensuremath{\mathsf{CFO}}\xspace}
\renewcommand{\leq}{\leqslant}
\renewcommand{\geq}{\geqslant}

\newcommand{\oicfo}{\ensuremath{<\!\!\text{-invariant } \cfo}\xspace}
\newcommand{\oifo}{\ensuremath{<\!\!\text{-invariant } \FO}\xspace}
\newcommand{\qeven}{\ensuremath{Q^{\text{even}}_{\text{atom}}}\xspace}

\newcommand{\con}[0]{\mathds C}
\newcommand{\ocon}[0]{\mathds O}
\newcommand{\econ}[0]{\mathds E}

\newcommand{\bc}[0]{\mathsf{SC}}

\newcommand{\mathcommand}[2]{\newcommand{#1}{\ensuremath{#2}\xspace}}

\mathcommand{\structdom}{A}
\mathcommand{\struct}{\mathcal A}
\mathcommand{\structbisdom}{A'}
\mathcommand{\structbis}{\mathcal A'}
\mathcommand{\tree}{\mathcal T}
\mathcommand{\treebis}{\tree'}

\NewDocumentCommand{\dtp}{O{h}O{t}O{a}O{\tree}}{\ensuremath{\text{dtp}_{#4}^{#1,#2}(#3)}\xspace}
\NewDocumentCommand{\dtpord}{O{h}O{t}}{\ensuremath{\prec^{#1,#2}}\xspace}

\NewDocumentCommand{\atp}{O{a}O{\struct}}{\ensuremath{\text{atp}_{#2}(#1)}\xspace}

\NewDocumentCommand{\boule}{O{k}O{a}O{\struct}}{\ensuremath{N_{#3}^{#1}(#2)}\xspace}
\NewDocumentCommand{\neigh}{O{k}O{a}O{\struct}}{\ensuremath{\mathcal N_{#3}^{#1}(#2)}\xspace}
\NewDocumentCommand{\tneigh}{O{k}O{a}O{\tree}}{\ensuremath{\triangle_{#3}^{#1}(#2)}\xspace}
\NewDocumentCommand{\tneightp}{O{k}O{a}O{\tree}}{\ensuremath{\blacktriangle_{#3}^{#1}(#2)}\xspace}

\NewDocumentCommand{\context}{O{k}O{a}O{(\struct,<)}}{\ensuremath{\con_{#3}^{#1}(#2)}\xspace}
\NewDocumentCommand{\ocontext}{O{k}O{a}O{(\struct,<)}}{\ensuremath{\ocon_{#3}^{#1}(#2)}\xspace}

\NewDocumentCommand{\Int}{O{\boule}O{(\struct,<)}}{\ensuremath{\text{Int}_{#2}(#1)}\xspace}
\NewDocumentCommand{\OInt}{O{\boule}O{(\struct,<)}}{\ensuremath{\text{Int}^+_{#2}(#1)}\xspace}

%% file: introduction.tex
\section{Introduction}

Logical languages like \emph{first-order logic} ($\FO$) are the main way to frame queries in finite model theory and database theory. By default, when writing a query in such a language, one only uses the relations appearing in the signature of the structures at hand. In practice however, structures do not exist in the void and are ordered (albeit in an arbitrary fashion) on tape or in memory, which only allow linear, one-dimensional representations of data. From this observation comes the natural question of whether our formal languages can use this additional information, provided the result of formula evaluation does not depend upon the precise order the input structure is endowed with.

To that end, we authorise the use of an extra binary symbol $<$ in the syntax of formulas,
which is guaranteed to be interpreted as a linear order on the elements of the structure. However, insofar as we should only define properties of the underlying unordered structures, it is asked that whether such a formula holds in the ordered structure does not depend on the choice of the order. Such a formula is called an order-invariant formula. Order invariance stems from the fields of descriptive complexity, in particular the quest for a logic for $\mathsf{PTIME}$,  database theory, where one studies generic queries that should depend only on the underlying unordered database, and finite model theory, where one often allows access to additional relations as technical devices \cite{DBLP:books/aw/AbiteboulHV95,DBLP:conf/dimacs/Vianu96}.  
Order-invariant first-order logic captures exactly this idea: formulas may use an auxiliary linear order, but their truth must be independent of which order is chosen \cite{DBLP:conf/csr/Schweikardt13}.

From the perspective of descriptive complexity, this makes order-invariant logics particularly interesting. Many classical logical characterisations of complexity classes are obtained over ordered structures, where a built-in order provides a canonical way to organise computation \cite{DBLP:books/daglib/0095988}. Studying order-invariant fragments is therefore a way to understand whether access to order genuinely increases expressive power, or whether order merely serves as a technical aid. In this sense, order-invariant logics belong to the broader programme of understanding definability on unordered finite structures \cite{DBLP:conf/csr/Schweikardt13,DBLP:conf/dimacs/Vianu96}.

A natural way to understand these logics is through their expressive power.
Access to a linear order provides a global enumeration of the elements, and
therefore it is not surprising that allowing formulas to access an order
can increase expressive power.
In the realm of finite structures, it has been shown that order-invariant first-order is more expressive than plain first-order logic \cite{DBLP:conf/csr/Schweikardt13}. Prominent separating examples are that of Rossman \cite{DBLP:conf/lics/Rossman03} (which only uses the successor relation induced by the order -- this restriction of order-invariant \FO is known as \emph{successor-invariant \FO}), and the one due to Gurevich (which can be found in \cite[Theorem~5.3]{DBLP:books/sp/Libkin04}). In all of the known examples proving the existence of order-invariant first-order formulas that have no first-order equivalent, the underlying structure is very complex from the eyes of graph structural theory. Hence, one may ask \emph{whether on tamer structures we can see a collapse of order-invariant \FO on plain \FO.} This is the case on some specific classes of trees and on classes of graphs with bounded tree-depth for which first-order logic and its order-invariant extension have the same expressive power.
For the former, the proof technique is to exhibit a chain of intermediate structures and intermediate orders that are pairwise similar (by which we informally mean that these structures agree on all the \FO-formulas up to some quantifier depth) \cite{DBLP:journals/jsyml/BenediktS09}. The latter is shown by proving that \MSO\footnote{Short for \emph{monadic second-order logic}, which is an extension of first-order logic where quantification over sets of vertices is allowed. In general, it is strictly more expressive than \FO.} and \FO have the same expressive power on classes of graphs of bounded tree-depth, and then proving the inclusion of order-invariant FO in MSO \cite{DBLP:journals/tocl/EickmeyerEH17}. To get this inclusion, starting from two FO-similar structures, one constructs orders that maintain this similarity. The same method is used in \cite{bednarczyk2025expressive} on \emph{first-order logic with two variables} ($\FO^2$) to show that order-invariant $\FO^2$ is included in $\FO$ on bounded degree structures. While similar in terms of methodology, the two mentioned works (\cite{DBLP:journals/tocl/EickmeyerEH17, bednarczyk2025expressive}) construct orders in widely different fashions, which are influenced by the very different restrictions (bounded tree-depth and bounded degree). See \Cref{table:summary_order_invariant_fo}, for a summary of known results on the expressive power of order-invariant \FO.

\begin{table}[ht]
\centering
\small
\renewcommand{\arraystretch}{1.15}
\begin{tabularx}{\textwidth}{>{\raggedright\arraybackslash}p{0.34\textwidth} X}
\toprule
\textbf{Setting} & \textbf{Result} \\
\toprule
Arbitrary finite structures
& Every order-invariant $\FO$ formula is Gaifman-local \cite{DBLP:journals/tocl/GroheS00}. \\
\midrule
Strings; trees
& Collapse to $\FO$ \cite{DBLP:journals/jsyml/BenediktS09}. \\
\midrule
Bounded-treewidth graphs; bounded-degree graphs
& Contained in MSO \cite{DBLP:journals/jsyml/BenediktS09}; \\
\midrule
Bounded-degree graphs&
Successor-invariant \FO collapses to \FO~\cite{DBLP:journals/lmcs/Grange21}. \\
\midrule

Bounded-treedepth graphs
& Collapse to $\FO$ \cite{DBLP:journals/tocl/EickmeyerEH17}. \\
\midrule
Planar graphs; more generally, decomposable structures
& Contained in MSO \cite{DBLP:conf/lics/ElberfeldFG16}. \\
\midrule
Hollow trees
& Collapse to $\FO$ \cite{DBLP:conf/csl/GrangeS20}. \\
\midrule
Classes of bounded degree structures
& Order-invariant $\FO^2$ is contained in $\FO$; moreover, the same bounded-degree containment also holds for the two-variable counting extension considered there \cite{bednarczyk2025expressive}. \\
\bottomrule
\end{tabularx}
\caption{Known results on the expressive power of invariant first-order logics.}
\label{table:summary_order_invariant_fo}
\end{table}


To summarise, the methods showing the collapse of order-invariant $\FO$ to plain $\FO$ are:
\begin{itemize}
	\item Showing the inclusion of order-invariant $\FO$ in $\MSO$ and an equivalence in the expressive power of $\MSO$ and $\FO$, on a class $\CCC$.
	\item Showing that for any two $\FO$-similar structures in $\CCC$, one can compute orders such that the expanded structures with said orders remain $\FO$-similar; the proof of why this implies the collapse of order-invariant \FO to plain \FO in given in \Cref{main theorem}.
\end{itemize}
Now, if one wants to extend the current results, bounded-degree graphs seem like a good starting point. We know that on these structures, order-invariant \FO is included in \MSO and successor-invariant \FO collapses to \FO~\cite{DBLP:journals/lmcs/Grange21}. It is therefore natural to ask whether, on such classes, order-invariant \FO collapse to \FO as well. As \MSO is strictly more expressive than \FO on graphs of bounded degree\footnote{\MSO can, for instance, express \textbf{connectivity}, see e.g. \cite[Chapter~4]{DBLP:books/sp/Libkin04}.}, the first method will not work.

In this work, we take a step towards settling the question of \emph{whether the expressive power of order-invariant first-order logic and plain first-order logic on graphs of bounded degree is the same.}

Although our main results concern the expressive power of order-invariant logics, the method used to prove them is also part of the contribution. The central step is to start from two bounded-degree structures that are sufficiently similar from the point of view of plain \FO, and to construct linear orders on them so that the resulting ordered structures remain indistinguishable by \CFO. In this sense, the proof shows how a global object, namely a linear order, can be added without destroying the relevant local similarity.
Understanding when such orders can be constructed, and how they have to be organised, gives structural information about the graphs under consideration. We believe that this order-construction method may also be useful in future attempts to compare order-invariant \FO with plain \FO on broader classes of structures. Thus, beyond the specific results proved here for \CFO, the paper develops a toolkit for studying order-invariant first-order logics.

Computing such good orders for \FO on bounded-degree graphs and even trees (where we already know that order-invariant \FO and plain \FO collapse) seems very hard and out of reach. Therefore, we tackle this problem by introducing a new fragment of first-order logic: \emph{cluster first-order logic} (\CFO).

Cluster first-order logic operates on ordered graphs, and variables of this logic are grouped in clusters. Variables of the same cluster can be compared in an unrestricted fashion (both in terms of edges and precedence in the order), while elements from different clusters can (almost) not be compared. To prevent us from grouping all the variables in the same cluster, variables of a cluster are guarded by the edge-relation, meaning that when a new variable from a cluster is quantified, it has to be adjacent to a previously quantified variable of the same cluster.
Up to this point, the logic is already as expressive as \FO on bounded-diameter graphs, as shown in \Cref{lemma:cfo_equals_fo_bounded_diameter}. To grant more expressive power to \CFO, we allow a single element of each cluster to be compared with the order to every element of one other cluster; the syntax of \CFO prevents this dependency to be circular. Order-invariant \CFO is the set of such sentences whose truth value do not depend on the order.


\emph{We show that order-invariant cluster first-order logic is included in \FO on classes of graphs of bounded degree (Theorem~\ref{main theorem}).} Let us emphasis again that we achieve this by computing good orders on similar bounded degree structures such that they remain similar after expanded by the orders.

We then prove that model-checking\footnote{\emph{Model-checking} is about determining whether a given formula $\varphi$ holds in a given graph $G$. It is customary to distinguish the size of $\varphi$ from the size of $G$ as the latter is often much larger than the former. As such, the complexity of model-checking problems is often looked at through the lense of parametrized complexity, and a model-checking algorithm will be considered satisfactory when the dependency on $G$ is polynomial. In case such an algorithm exists, the model-cheking problem is said to be fixed-parameter tractable.} of an order-invariant cluster first-order formula on classes of bounded degree can be done in fixed-parameter tractable time. Model-checking results for order-invariant and successor-invariant \FO such as those in \cite{EickmeyerEtAl20, DBLP:conf/csl/EickmeyerK16,DBLP:conf/lics/HeuvelKPQRS17} proceed by computing nice orders or successor relations such that when a graph from those classes is expanded to a structure with those relations, the expanded class remains a class for which a fixed-parameter tractable algorithm exists. This way, model-checking an order- or successor-invariant \FO formula is reduced to model-checking a plain \FO formula to a class of structures with order or successor relations. This method shall not work for order-invariant \FO on bounded degree graphs. Indeed, bounded-degree classes of graphs are incomparable with bounded twin-width classes of graphs, and it has been shown that model-checking of plain \FO formulas on ordered structures is fixed-parameter tractable on a class $\CCC$ if and only if $\CCC$ has bounded twin-width \cite{BonnetEtAl24}. Therefore, it is impossible to order every bounded-degree class of graphs so as to get a tame class of structures, in the sense that \FO model-checking on this class is fixed-parameter tractable\footnote{More precisely, there exists a class of bounded degree graphs such that expanded with any order relations, the expanded class will be monadically independent.}. A concrete example of such a class is cubic graphs.

Accordingly, we had to approach the model-checking of order-invariant \CFO in a totally different way. We know that on a class of bounded degree, every sentence of \oicfo has an equivalent \FO-sentence: an idea could be to translate \oicfo to \FO, and then to rely on the fixed-paramater tractability of \FO on graphs of bounded degree~\cite{DBLP:journals/mscs/Seese96}. However, we do not know such a translation to be computable. Instead, we use the structural tools developed for the expressive-power result. More precisely, we show that on bounded-degree graphs, the truth of an \oicfo sentence depends only on finitely many local configurations. Since only boundedly many such configurations can occur, model-checking reduces to checking the formula on a finite summary of the ordered graph. We then show that this summary can itself be computed in fixed-parameter tractable time, which yields a fixed-parameter tractable model-checking algorithm on bounded-degree graphs.

As a last result, we adapt Gurevich's example separating order-invariant \FO from \FO to get an order-invariant \CFO property that falls outside the scope of \FO.

To summarise:
\begin{itemize}
	\item \oicfo is included in plain \FO on bounded-degree classes of graphs (\Cref{main theorem}), and model-checking is fixed-parameter tractable on such classes (\Cref{th:model-checking}), and
	\item it is strictly more expressive than \FO in general (\Cref{theorem: separating example}).
\end{itemize}

\input{overview}


%% file: overview.tex
\section{Preliminaries}
\subsection{Graphs, general notations.} In this setting, we consider \emph{loopless, undirected, coloured simple graphs}. Such graphs can be formally represented as relational structures over the vocabulary $\sigma = \{E, P_1, \dots, P_m\}$, where $E$ is a binary predicate representing edges (whose interpretation is always a symmetric and loopless binary relation) and each $P_i$ is a unary predicate representing a vertex colour. We will also consider ordered graphs, by which we mean structures over vocabulary $\sigma\cup\{<\}$ where $<$ is a binary predicate, interpreted as a linear order on the elements. Graphs and structures are denoted $\mathcal A,\mathcal B,\dots$ while their sets of elements are denoted $A, B,\dots$; if $P\subseteq A$, we write $\mathcal A_{|P}$ for the substructure of $\mathcal A$ induced by $P$. An isomorphism between two structures is a bijection between their element set that respects all predicates of $\sigma$ in both directions.

For standard terminology and further background on the underlying structural properties, we refer to \cite{DBLP:books/sp/Libkin04}
 for a more comprehensive treatment of finite model theory. 

 \subsection{First-order logic.} We use standard first-order logic \FO over relational vocabularies. If $\tau$ is such a vocabulary, formulas over $\tau$ are built from atomic formulas $R(x_1,\dots,x_n)$, for $R\in\tau$, and equalities $x=y$, using Boolean connectives and quantification over elements of the structure. When we want to make the vocabulary explicit, we write $\FO(\tau)$ for first-order logic over $\tau$.

Given a formula $\varphi\in\FO(\tau)$ with free variables in $X$, a $\tau$-structure $\mathcal A$, and a valuation $\nu:X\to A$, we write $(\mathcal A,\nu)\models\varphi$ if $\varphi$ holds in $\mathcal A$ under $\nu$. If $X=\emptyset$, then $\varphi$ is a sentence and we simply write $\mathcal A\models\varphi$.

The \emph{quantifier rank} of an $\FO$ formula is the maximum nesting depth of its quantifiers. We write $\FO[k]$ for the set of $\FO$ sentences of quantifier rank at most $k$.

\subsection{Cluster first-order logic}
We define \emph{cluster first-order logic} (\CFO), a syntactic fragment of \FO over ordered,
coloured graphs. Variables are organised into \emph{clusters} indexed by words
$w \in \Sigma^*$ over a fixed infinite alphabet $\Sigma$.
Inside each cluster, variables are enumerated
$x^0_w, x^1_w, x^2_w, \ldots$.
Except for $x^0_w$, variables may only be introduced by \emph{guarded}
quantification, enforcing adjacency to a previously introduced variable of the
same cluster. Moreover, $E$ and $=$ may only relate variables from the same
cluster. The order relation $<$ may additionally compare a variable from cluster
$w$ with the first variable $x^0_{w\alpha}$ of a child cluster $w\alpha$, where $\alpha\in\Sigma$.

\noindent
\textbf{Indices and variables.}
Fix an infinite alphabet $\Sigma$.
For a finite set of indices $S \subseteq \Sigma^* \times \mathbb{N}$, define
$X_S := \{x^i_w \mid (w,i) \in S\}.$
We say that $S$ is \emph{valid} if for all $w \in \Sigma^*$, $\alpha \in \Sigma$,
and $i \in \mathbb{N}$:
\begin{itemize}
  \item if $(w,i) \in S$, then $(w,j) \in S$ for every $j < i$;
  \item if $(w\alpha,0) \in S$, then $(w,0) \in S$.
\end{itemize}
\noindent
\textbf{Syntax.}
Let $\sigma := \{E, P_1, \ldots, P_N\}$ where $E$ is binary and each $P_\ell$ is
unary, and consider the vocabulary $\sigma \cup \{<\}$ where $<$ will be interpreted
as a linear order. For every valid set $S$, we define $\CFO_S[k]$ (formulas of
quantifier rank at most $k$) inductively.

For $k=0$, $\CFO_S[0]$ is the smallest set containing the following atomic
formulas and closed under Boolean connectives:
\begin{itemize}
  \item $P_\ell(x^i_w)$ for $x^i_w \in X_S$,
  \item $E(x^i_w,x^j_w)$ and $x^i_w \sim x^j_w$ for $x^i_w,x^j_w \in X_S$ and
        ${\sim} \in \{=,<\}$,
  \item $x^0_{w\alpha} \sim x^i_w$ for $x^0_{w\alpha}, x^i_w \in X_S$ and
        ${\sim} \in \{=,<\}$.
\end{itemize}
For $k+1$, $\CFO_S[k+1]$ is the smallest set closed
under Boolean connectives that contains all formulas of the
following forms (where $\psi \in \CFO_{S'}[k]$):
\begin{itemize}
  \item Root introduction (only when $S=\emptyset$):
  $\exists x^0_\epsilon \, \psi
 \quad\text{where } S'=\{(\epsilon,0)\}.$
  \item Opening a child cluster:
  $\exists x^0_{w\alpha} \, \psi
  \quad\text{where } (w,0)\in S,\ (w\alpha,0)\notin S,\ \alpha\in\Sigma,
  \text{ and } S' = S \cup \{(w\alpha,0)\}.$
  \item Guarded continuation inside a cluster:
  $\exists x^i_w \bigl(E(x^i_w,x^j_w) \wedge \psi\bigr),$
  where $(w,j)\in S$ and $i$ is the smallest index such that $(w,i)\notin S$,
  and $S' = S \cup \{(w,i)\}$.
\end{itemize}
Finally, we let $\CFO_S := \bigcup_{k \in \mathbb{N}} \CFO_S[k],$
omit $S$ when $S=\emptyset$, and write simply $\CFO$.
\noindent

\textbf{Semantics.}
A formula $\psi \in \CFO_S$ is interpreted as a first-order formula with free
variables in $X_S$. For an ordered $\sigma$-structure $(\mathcal A,<)$ and a
valuation $\nu : X_S \to A$, satisfaction is defined as usual and denoted
$(\mathcal A,<,\nu) \models \psi$.

\begin{example}
  Formula $\varphi_1$ detects a triangle within a specific cluster $\alpha$ ($\alpha\in\Sigma$).
\begin{equation*}
\varphi_1 :=
\exists x_{\alpha}^1 \Bigl(
  E(x_{\alpha}^1, x_{\alpha}^0)
  \land
  \exists x_{\alpha}^2 \bigl(
    E(x_{\alpha}^2, x_{\alpha}^1)
    \land
    E(x_{\alpha}^2, x_{\alpha}^0)
  \bigr)
\Bigr).
\end{equation*}
This formula belongs (among other) to $\CFO_{\{(\epsilon,0),(\alpha,0)\}}$. In fact, $\varphi_1$ belongs to $\CFO_S$ for any valid set $S$ of indices containing $(\alpha,0)$ but not $(\alpha,1)$. The sub-formula $\exists x^2_{\alpha} (E(x^2_{\alpha}, x^1_{\alpha}) \wedge E(x^2_{\alpha}, x^0_{\alpha}))$ belongs to $\CFO_{\{(\epsilon,0),(\alpha,0),(\alpha,1)\}}$.

Formula $\varphi_2\in\CFO_{\{(\epsilon,0)\}}$ asserts the existence of a neighbour $x^1_\epsilon$ of $x^0_\epsilon$ such that the interval between $x^0_\epsilon$ and $x^1_\epsilon$ (by which we mean all the elements of the structure that are between $x_\epsilon^0$ and $x_\epsilon^1$ with respect to $<$) contains a vertex belonging to some triangle:
\begin{equation*}
\varphi_2 :=
\exists x_\epsilon^1 \Bigl(
  E(x_\epsilon^1, x_\epsilon^0)
  \land
  \exists x_{\alpha}^0 \Bigl(
    (x_\epsilon^0 < x_{\alpha}^0 \land x_{\alpha}^0 < x_\epsilon^1)
    \lor
    (x_\epsilon^1 < x_{\alpha}^0 \land x_{\alpha}^0 < x_\epsilon^0)
  \Bigr)
  \land
  \varphi_1
\Bigr).
\end{equation*}
Crucially, note that it is not possible in \CFO to check whether the two other vertices of the triangle are in the interval between $x_\epsilon^0$ and $x_\epsilon^1$. In $\CFO$, for two variables in the same cluster $w$, we can refer to at most one vertex at a time that lies in the interval between them and is not in the same cluster $w$, by using the first variable of a new cluster $w\alpha$ ($\alpha\in\Sigma$).
\end{example}
As can be expected, the obligation for new variables of a given cluster to be guarded via $E$ to a previous variable does not hinder the expressive power of \CFO on graphs of small diameters:

\begin{remark}
  \label{lemma:cfo_equals_fo_bounded_diameter}
  Let $\mathcal C$ be a class of ordered graphs such that the underlying graph of each $(\mathcal A,<)\in\mathcal C$ has diameter at most $\delta$.

  Then $\FO(\sigma\cup\{<\})$ and \CFO have the same expressive power on $\mathcal C$. In other words, every subclass of $\mathcal C$ definable by a first-order sentence over the ordered vocabulary $\sigma\cup\{<\}$ is definable in \CFO, and conversely.
\end{remark}

\begin{proofsketch}
  One direction needs no proof, insofar as \CFO is a fragment of $\FO(\sigma\cup\{<\})$. For the other direction, consider a sentence $\varphi\in\FO(\sigma\cup\{<\})$. Without loss of generality, using the equivalence between $\forall x\psi$ and $\neg\exists x\neg\psi$, we can thus assume $\varphi$ contains only existential quantifiers.

  We construct a sentence $\widehat\varphi\in\CFO$ equivalent to $\varphi\in\FO(\sigma\cup\{<\})$ as follows, only using cluster $\epsilon$.
  First, we prepend $\varphi$ with the existential quantification $\exists x^0_\epsilon$. Then, moving from the outside in, and assuming the variables $x_\epsilon^0,\dots,x_\epsilon^i$ have been used so far in the construction of $\widehat\varphi$, we replace subformula $\exists x \psi$ via \[\bigvee_{j=0}^{\delta}\exists x_\epsilon^{i+1}E(x_\epsilon^{i+1},x_\epsilon^0)\land \exists x_\epsilon^{i+2}E(x_\epsilon^{i+2},x_\epsilon^{i+1})\land\cdots\land\exists x_\epsilon^{i+j}E(x_\epsilon^{i+j},x_\epsilon^{i+j-1})\land\psi\,.\]
All variables introduced in the translation belong to the same cluster, so all atomic relations of the ordered vocabulary appearing in $\varphi$, including order atoms, are allowed in the resulting \CFO formula.
Since in the underlying graphs of structures from $\mathcal C$, every possible witness to $\exists x \psi$ is at distance at most $\delta$ from the vertex interpreting $x_\epsilon^0$, $\varphi$ and $\widehat\varphi$ are equivalent on $\mathcal C$.
\end{proofsketch}

A sentence $\varphi\in\CFO$ is \emph{order-invariant} (or $<$-invariant) if its truth value is independent of the specific order chosen on a graph, i.e. if for any finite graph $\mathcal A$, and any two linear orders $<_1$ and $<_2$ on its vertex-set $A$,
\[
(\mathcal A, <_1) \models \varphi \iff (\mathcal A, <_2) \models \varphi\,.
\]
In this case, we write $\mathcal A\models\varphi$. We write \oicfo for the set of all order-invariant \CFO-sentences.
This ensures the logic defines properties inherent to the underlying graph $\mathcal A$ rather than an arbitrary arrangement of its elements.

Given a logic $\mathcal{L}\in\{\FO,\CFO,\oicfo\}$ and a set $X$ of variables, two structures $\mathcal A$ and $\mathcal B$ (specifically, ordered graphs in the case of \CFO and graphs in the case of \oicfo), together with valuations $\nu_A:X\to A$ and $\nu_B:X\to B$, are said to be $k$-\emph{equivalent} (or $k$-\emph{similar}) with respect to $\mathcal{L}$, denoted $(\mathcal A,\nu_A) \equiv^k_\mathcal{L}(\mathcal B,\nu_B)$, if they satisfy exactly the same formulas of $\mathcal L$ with quantifier rank at most $k$.
If $X=\emptyset$, we omit the valuations and write $\mathcal A \equiv^k_\mathcal{L} \mathcal B$.

To reason about indistinguishability in \CFO, we use a game characterisation called the
\emph{cluster Ehrenfeucht--Fra\"{i}ss\'{e} game} (CEF), which captures exactly the expressive
power of \CFO. The game is played for $k$ rounds between two players, \emph{Spoiler}
and \emph{Duplicator}, on two ordered graphs. In each round, Spoiler selects an
element in a \CFO-consistent way, either opening a new cluster
or introducing a new neighbour inside an existing cluster, and Duplicator must
respond in the other graph with a matching choice that preserves adjacency,
equality, and the restricted order comparisons allowed in \CFO. Duplicator aims to
maintain a partial isomorphism respecting the cluster discipline, while Spoiler tries
to break this correspondence. Duplicator wins if they can maintain these conditions for
all $k$ rounds. The precise connection with \CFO-equivalence is the following.
\begin{restatable}[CEF captures $\cfo$]{theorem}{CEFcapturesCFO}
\label{th:CEF_captures_CFO}
Let $(\mathcal A,<_A)$ and $(\mathcal B,<_B)$ be two ordered graphs. Then $(\mathcal{A},<_A)\equiv^k_{\cfo}(\mathcal{B},<_B)$ if and only if Duplicator has a winning strategy in the $k$-round cluster EF game between $(\mathcal A,<_A)$ and $(\mathcal B,<_B)$.
\end{restatable}

Game characterisations of logical equivalence are classical in finite model theory,
and typically follow a common template. We use such a game as a technical tool in
our proofs to establish indistinguishability between ordered graphs.

We then move to showing our main result.
\section{Main result}
Order invariance means that formulas are allowed to use an additional linear
order $<$ on the vertices as a technical aid, but the property they define
should not depend on which particular order is chosen. In other words, even
though the formula may refer to $<$, its truth must be the same for every
possible ordering of a given graph.

When comparing two logics on a class of graphs, we say that
one is included in the other if every property definable on that class by the first logic can be expressed as well by the second one. Two graphs $\mathcal A$ and $\mathcal B$ are $k$-equivalent with respect to a logic $\mathcal L$, denoted $\mathcal A\equiv^k_{\mathcal L}\mathcal B$, if they satisfy exactly the same $\mathcal L$-formulas of quantifier rank at most $k$.

We can now formally state the main result of this work.
\begin{restatable}{theorem}{Maintheorem}
\label{main theorem}
Let $\mathcal{C}$ be a class of graphs of bounded degree. Order-invariant \CFO is included in \FO on $\mathcal C$.
\end{restatable}

\textbf{Proof sketch.}
To prove the inclusion $<$-invariant $\CFO \subseteq \FO$ on $\CCC$, we compare how well the two logics can distinguish graphs. For a fixed quantifier rank $k$, it suffices to find a function $f(k)$ such that whenever two graphs $A$ and $B$ cannot be distinguished by any $\FO$ sentence of rank $f(k)$, they also cannot be distinguished by any $<$-invariant $\CFO$ sentence of rank $k$. Intuitively, this means that $\FO$ is at least as expressive: if $\FO$ sees no difference, then $\CFO$ sees none either.
Another way of looking at this is to notice that both logics partition $\CCC$ into finitely many indistinguishability classes, each one definable in that logic, and each $\oicfo[k]$ class is a union of $\FO[f(k)]$ classes. Since a formula of $\oicfo[k]$ simply covers a union of $\oicfo[k]$ classes, it can in turn be defined in $\FO[f(k)]$.

Finally, because $<$-invariant formulas do not depend on the chosen order, it is enough to equip $A$ and $B$ with suitable orders $<_A$ and $<_B$ and show $(A,<_A) \equiv^k_{\CFO} (B,<_B)$ whenever $A \equiv^{f(k)}_\FO B$. 
It therefore remains to show that whenever two graphs are sufficiently similar,
we can choose linear orders for each of them such that, after equipping the
graphs with these orders, they become indistinguishable in $\CFO$. In other
words, by picking the right orders, we can reduce the problem to $\CFO$
equivalence on ordered graphs.

We prove this by introducing a special family of orders, called $(k,F)$-orders,
formally defined in \Cref{definition:order}, which satisfy several useful
structural properties. First, every graph that is sufficiently rich admits such
an order (\Cref{good graphs admit good orders}). Next, any two graphs that are
similar enough admit corresponding orders that are themselves similar
(\Cref{transfer lemma}). Finally, these orders allow us to establish
$\CFO$-equivalence: by equipping the two graphs with their respective orders,
Duplicator obtains a winning strategy in the CEF game. This strategy is given
explicitly in \Cref{main technical theorem}.
Our construction relies on two complementary ways of describing the local
structure of a bounded-degree graph: one tailored to $\FO$ and one to $\CFO$.
For $\FO$, we use \emph{neighbourhood types}. Fix a radius $r$. The
$r$-neighbourhood type of a vertex $v$ is the isomorphism type of its
radius-$r$ neighbourhood $N_r(v)$. Since $\FO$ formulas of bounded quantifier
rank can only explore bounded neighbourhoods, vertices with the same
$r$-neighbourhood type are indistinguishable to $\FO$. Hence, on bounded-degree
graphs, only finitely many such types can occur, and from the point of view of
$\FO$, a graph is essentially determined by how many vertices realise each type.

For $\CFO$, we use \emph{contexts}, which play an analogous role for the
clustered, order-based setting. Contexts capture exactly the local configurations
that $\CFO$ can distinguish. Thus, neighbourhood types describe what $\FO$ can
see, while contexts describe what $\CFO$ can see, and on bounded-degree graphs
both notions give rise to only finitely many possibilities and are
functionally equivalent in their counts (\Cref{corollary: all contexts in finite}).
Consequently, if two graphs are $\FO$-equivalent for a sufficiently large
quantifier rank, they realise the same neighbourhood types with the same
multiplicities, and therefore also admit matching contexts. This correspondence
allows us to align their vertices and construct matching $(k,F)$-orders, which is
the key step toward proving $\CFO$-equivalence.

Since contexts are central to the proof, we briefly describe them here and refer to
\Cref{def:context} for the formal definition. A $k$-context consists of the
$k$-neighbourhood of a vertex equipped with a linear order that partitions it
into intervals, together with the $(k-1)$-contexts occurring within these
intervals, thereby recording how smaller local configurations are arranged along
the order.

In bounded-degree graphs, if a $k$-neighbourhood type occurs sufficiently many
times--say at least $g(k,r,d,m)$ times, where this means that there are
$g(k,r,d,m)$ distinct vertices with that type, for a function depending only on
$k$, the degree bound $d$, and the parameters $r$ and $m$--then there exist $m$
such vertices whose pairwise distances are greater than $r$. This property extends to several neighbourhood types simultaneously. In other
words, if each of a fixed collection of neighbourhood types occurs sufficiently
many times, where the required number depends only on the neighbourhood depth,
the maximum degree, the desired pairwise distance, and the number of types under
consideration, then we can select many vertices of these types whose pairwise
distances are all greater than the prescribed bound. Thus, we can obtain many
mutually far-apart occurrences across all the chosen types.


Now let $F$ be a fixed collection of $k$-neighbourhood types.  
We say that a graph is \emph{rich for $(k,F)$} if each type in $F$ occurs
sufficiently many times, namely at least a threshold $M$.  This value $M$ is
precisely the bound discussed in the previous paragraph: it depends only on $k$
and the degree bound and guarantees many well-separated occurrences of each
type.  Its exact numerical value is computed in
\Cref{main technical theorem}, but it is not important for this introductory
discussion.  What matters is the consequence that every type in $F$ appears more
than $M$ times, which allows us to select many occurrences that are pairwise far
apart and hence can be used independently to build a structured order.

We construct a $(k,F)$-order for a graph that is rich for $(k,F)$ as follows.
\begin{enumerate}
    \item We first place all vertices whose $k$-neighbourhood types do not belong
    to $F$ at the beginning of the order.  This initial segment is called the
    \emph{extremity} and is denoted by $X$.

    \item Next, we introduce $2k^2$ segments distributed along the order, called
\emph{universal segments}.  In each such segment we place one copy of every
context that can be constructed using only vertices whose $k$-neighbourhood types
belong to $F$.  These segments are called universal because each of them contains
representatives of all such contexts.  Recall that a context is simply a small
ordered local configuration, consisting of a few vertices together with their
bounded neighbourhoods arranged according to the order.  Constructing a context
therefore amounts to selecting appropriate vertices with the required
neighbourhood types and arranging them in the prescribed order.

The role of the threshold $M$ becomes crucial here.  Since every type in $F$
occurs many times, we can choose occurrences that are pairwise far apart.  This
allows us to build the contexts inside different universal segments using
disjoint sets of vertices and neighbourhoods, ensuring that they do not
interfere with one another.
    \item Between consecutive universal segments we insert
    \emph{neighbouring segments}, which contain the vertices in the neighbourhood
    of the previous universal segment.  Because the chosen contexts are
    well-separated, these neighbourhoods do not overlap and the construction can
    be carried out independently around each segment.

    \item Finally, all remaining vertices are placed in a central segment called
    the \emph{jungle}.
\end{enumerate}
Thus the order has a fixed global structure consisting of the extremity,
several universal segments containing all contexts one can construct with the neighbourhood types in $F$, their neighbouring segments,
and the jungle. This global structure comes with an order $\prec$ which is the order between the segments. More precisely, when comparing elements $a,b$ from different segments, $b$ is greater than $a$ if it belongs to a segment to the right of the segment of $a$. This order is depicted in \Cref{fig:ktl-order}.

Within each universal segment we fix an arbitrary but consistent order on the
contexts.  Concretely, if $c_1,\dots,c_\ell$ are all contexts that can be
constructed using neighbourhood types from $F$, we place their chosen copies in
the order $c_1 \prec' \cdots \prec' c_\ell$, meaning that every vertex belonging to
$c_i$ precedes every vertex belonging to $c_j$ whenever $i<j$, in this short introduction on the order by $a\in c_i$ we mean that the element $a$ belongs to the occurrence of the context $c_i$.  On the extremity
$X$, the neighbouring segments, and the jungle $J$, we choose arbitrary orders.
For a graph $\mathcal A$, we denote the resulting order by $<_A$
and call it a $(k,F)$-order. The existence of such an order relies precisely on the richness
assumption, see \Cref{good graphs admit good orders}.  
Recall that each context $c$ comes with a built-in order $<_{c,\mathcal A}$. Therefore, the order $<_A$ compares elements as follows.

\[
a <_A b
\quad \Longleftrightarrow \quad
\begin{cases}
\text{segment of }a \prec \text{segment of }b, & \text{if } \text{segment of }a \neq \text{segment of }b,\\[0.4em]
c_i\prec' c_j, & \text{if } a,b\text{ belong to the same universal segment},\\& a\in c_i,\ b\in c_j,\ \text{and } c_i\neq c_j,\\[0.4em]
a<_{c_i,\mathcal A}b, &\text{if } a,b\text{ belong to the same universal segment},
\\& \text{and } a,b \text{ belong to the same context } c_i,\\[0.4em]
a<b& \text{if } a,b\text{ belong to }X\text{ or }J\text{ or a neighbouring segment}\\[0.4em]
&\text{and }< \text{is the order in that segment.}
\end{cases}
\]


Moreover, if two graphs $\mathcal A$ and $\mathcal B$ are rich for the same
$(k,F)$ and have matching counts for the neighbourhood types outside $F$, we
first construct a $(k,F)$-order $<_A$ on $\mathcal A$ as above and then transfer
this order to $\mathcal B$, obtaining an order $<_B$ with the same segment
decomposition (see \Cref{transfer lemma}).  Finally, we show that Duplicator has a
winning strategy in the CEF game on $(\mathcal A,<_A)$ and $(\mathcal B,<_B)$ (See \Cref{main technical theorem}),
yielding $(\mathcal A,<_A)\equiv^k_{\CFO}(\mathcal B,<_B)$ and completing the
proof of \Cref{main theorem}.

\definecolor{mygreen}{RGB}{36, 200, 100}
\definecolor{myyellow}{RGB}{220, 220, 30}
\usetikzlibrary{patterns, decorations.pathreplacing, shapes}

\begin{figure*}[!ht]
\label{the order}
  \centering
  \begin{tikzpicture}[scale=.98]

    \newcommand{\topy}{.5}
    \newcommand{\boty}{.5}

    \foreach \b/\e in { 3/4, 5/6, 6.5/7.5, 8/9, 10/11,, 11.4/12.4, 13/14}
    \draw[pattern=north west lines, pattern color=green!50] (\b,-\boty) rectangle (\e,\topy);

    \foreach \b/\e in {2/3, 4/5, 14/15}
    \draw[pattern=dots, pattern color=black!20] (\b,-\boty) rectangle (\e,\topy);

    \foreach \b/\e in {1/2, 9/10}
    \draw[pattern=crosshatch dots, pattern color=red!50] (\b,-\boty) rectangle (\e,\topy);

    \foreach \a in {6.3, 7.8,11.2,12.7}
    \node at (\a,0) {\small $\cdots$};

    \node at (1.5,0) {\small $X$}; 
    \node at (2.5,0) {\small $LU_1$}; 
    \node at (3.5,0) {\small $LN_1$}; 
    \node at (4.5,0) {\small $LU_2$}; 
    \node at (5.5,0) {\small $LN_2$};

    \node at (7,0) {\small $LN_{k^2+1}$}; 
    \node at (8.5,0) {\small $LN_{2k^2}$};
    \node at (9.5,0) {\small $J$}; 
    \node at (10.5,0) {\small $RN_{2k^2}$}; 
    \node at (13.5,0) {\small $RN_1$}; 
    \node at (12,0) {\small $LN_{k^2+1}$};
    \node at (14.5,0) {\small $RU_1$};

       \foreach \b/\e in {
       1.66/3.33, 3.66/5.16, 2.66/3.33, 4.66/5.16, 5.33/6.2, 6.3/7, 7.16/7.83, 
      7.9/8.66, 8.83/9.33, 9.66/10.33, 10.66/11.2, 11.3/11.9, 12/12.6, 12.7/13.5, 13.6/14.5
    }
    \draw[-] (\b,\boty) edge[out=70,in=110,-] (\e,\topy);

    \draw [decorate,decoration={brace,amplitude=3pt}] (6.5,-0.6) -- (2,-0.6) node [black,midway,yshift=-0.4cm] {$a$};
    \draw [decorate,decoration={brace,amplitude=3pt}] (9,-0.6) -- (6.5,-0.6) node [black,midway,yshift=-0.4cm] {$b$};
    \draw [decorate,decoration={brace,amplitude=3pt}] (12.4,-0.6) -- (10,-0.6) node [black,midway,yshift=-0.4cm] {$c$};
    \draw [decorate,decoration={brace,amplitude=3pt}] (15,-0.6) -- (12.4,-0.6) node [black,midway,yshift=-0.4cm] {$d$};
      
  \end{tikzpicture}
  \caption{
 (a) $k^2$ universal and neighbouring left segments, (b) $k^2$ only left neighbouring segments, (c) $k^2$ only right neighbouring segments, and (d) $k^2$ universal and neighbouring right segments. On the left side, $LN_i$ contains the neighbours of the vertices in $LU_i\cup LN_{i-1}$ for $2\leq i\leq k^2$, only $LN_{i-1}$ for $k^2+1\leq i$, and $LN_1:= N(LU_1\cup X)$. The same patterns appear on the right side. The extremity $X$ contains all those vertices whose $k$-neighbourhood types does not belong to $F$. The Jungle $J$, is filled with the remaining vertices after all other segments are filled.
 \  }
  \label{fig:ktl-order}
\end{figure*}
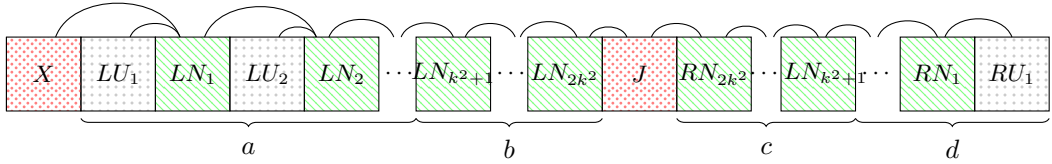

\section{Model checking}

In this section, we tackle the \emph{model-checking} problem for \oicfo. This problem amounts, given a graph and a sentence from \oicfo, to deciding whether this sentence holds in the input graph. Since both input are asymetric in size -- the graph is expected to be much larger than the formula -- it is important for the time complexity to be polynomial in the size of the graph, but less crucial if the dependency in the length of the formula is expensive. This corresponds to deciding whether this model-checking problem is \emph{fixed-parameter tractable} (FPT), where the size of the formula is taken as parameter.

We prove that when the degree is bounded, the model-checking problem for \oicfo is indeed FPT. Altough \FO is known to have an FPT model-checking problem when the degree is bounded~\cite{DBLP:journals/mscs/Seese96}, this result is not an immediate consequence of the inclusion of \oicfo in \FO as we do not know the translation to be computable.

Instead, we use again the tools that were developed in the previous section, namely the notion of contexts, to prove that whether a sentence holds in a given graph depends only on the set of contexts occurring in some ordering of this graph. Note that we must consider a slight extension of the notion of $k$-contexts, in the form of \emph{$k$-outer contexts}, where the $(k-1)$-contexts occurring in the leftmost and rightmost intervals are considered, while they are disregarded in regular $k$-contexts. Given that the number of $k$-outer contexts is bounded, the model-checking can then take place on some data whose size is bounded independently of the size of the input graph.

We then show that the set of $k$-outer contexts realised in an ordered graph can be computed in FPT. When combined, these two results provide an FPT algorithm for the model-checking problem on classes of bounded degree.

\begin{restatable}{theorem}{modelcheckingtheorem}
\label{th:model-checking}
The model-checking problem for \oicfo on a class $\mathcal C$ of bounded degree is fixed-parameter tractable.
\end{restatable}

\section{Separating example}

We conclude with a result showing that the bounded-degree hypothesis is crucial for \oicfo to be included in \FO:
\begin{restatable}{theorem}{seperatingexample}
\label{theorem: separating example}
 There exist properties definable in \oicfo that are not definable in \FO.
\end{restatable}The separation follows from a classical example due to Gurevich. The property
$\mathsf{QEven}$ states that a structure is isomorphic to a Boolean algebra
$(2^X,\subseteq)$ where $|X|$ is even. This property is known not to be
$\FO$-definable by an Ehrenfeucht--Fra\"iss\'e argument, but it is definable in
order-invariant first-order logic. We show that the same construction can already be carried
out in order-invariant $\CFO$: the order is used to isolate a single connected
component, and the clustered quantification then simulates the original
definition. Hence $\mathsf{QEven}$ is definable in order-invariant $\CFO$ but not
in $\FO$.


%% file: Preliminaries.tex
\begin{toappendix} 
\section{Preliminaries}   
%
\subsection{Cluster first-order logic.}

Before giving a formal definition, let us describe briefly \emph{cluster first-order logic} (\CFO), a fragment of \FO on ordered graphs. We start by fixing an infinite alphabet $\Sigma$. In a \CFO-formula, variables are indexed as $(x^i_w)_{w \in \Sigma^*}^{i \in \mathbb{N}}$ and organized into sets called \emph{clusters}: a cluster regroups all variables sharing a subscript $w$. The quantification of a variable $x^j_w$ with $j > 0$ is guarded: its interpretation has to be edge-adjacent to the interpretation of a previously introduced variable within the same cluster. Furthermore, in \CFO, whether the interpretation of two variables are related by $E,<$ or $=$ can only be checked if they belong to the same cluster, with one exception: comparisons using the linear order $<$ are also permitted between the first variable $x_{w\alpha}^0$ of cluster $w\alpha$ (for $w\in\Sigma^*$ and $\alpha\in\Sigma$) and variables of cluster $w$. With this intuition in mind, we can give a formal definition of \CFO.

Let $\Sigma$ be an infinite alphabet, and $S\subseteq\Sigma^*\times\N$ be a finite set of indices. We say that $S$ is \emph{valid} if the following holds for all $w\in\Sigma^*,\alpha\in\Sigma$ and $i\in\N$:
\begin{itemize}
\item $(w,i)\in S$ entails $(w,j)\in S$ for every $j<i$,
\item $(w\alpha,0)\in S$ entails $(w,0)\in S$.
\end{itemize}

In the remainder of this paper, given a valid set of indices $S$, $X_S$ denotes the set of variables $\{x_w^i:(w,i)\in S\}$.

Fixing a vocabulary $\sigma:=\{E,P_1,\dots,P_N\}$ where $E$ is binary and the $P$s are unary predicates, as well a an additional binary predicate $<$, we define $\CFO_S[k]$ by induction on $k$ as follows:
\begin{itemize}
\item $\CFO_S[0]$ contains atomic fomulas of the form $P(x_w^i)$, $E(x_w^i,x_w^j)$, $x_w^i\sim x_w^j$ and $x_{w\alpha}^0\sim x_w^i$ for $x_w^i,x_w^j,x_{w\alpha}^0\in X_S$ and ${\sim}\in\{=,<\}$, and is closed under boolean combinations;
\item $\CFO_S[k+1]$ contains formulas of the form
  \begin{itemize}
  \item $\exists x_\epsilon^0\psi$, where $\psi\in\CFO_{\{(\epsilon,0)\}}$ (if $S=\emptyset$)
  \item $\exists x_{w\alpha}^0 \psi$, where $\alpha\in\Sigma$, $(w\alpha,0)\notin S$, $(w,0)\in S$ and $\psi\in\CFO_{S\cup\{(w\alpha,0)\}}[k]$
  \item $\exists x_w^i E(x_w^i,x_w^j)\land \psi$, where $(w,j)\in S$, $i=\min\{l:(w,l)\notin S\}$ (because $S$ is valid, $j<i$) and $\psi\in\CFO_{S\cup \{(w,i)\}}[k]$
  \end{itemize}
  and is closed under boolean combinations.
\end{itemize}
We write $\CFO_S:=\bigcup_{k\in\N}\CFO_S[k]$, and omit the subscript $S$ when $S=\emptyset$. In particular, $\CFO:=\CFO_\emptyset$. 

Note that $\CFO_S$ is a subset of \FO. As such, we endow this logic with the usual semantics: given $\psi\in\CFO_S$, a $\sigma\cup\{<\}$-structure $(\mathcal A,<)$ (which is an ordered graph whose vertices are coloured by predicates $P_1,\dots,P_N$) and a valuation $\nu:X_S\to A$, we write $(\mathcal A,<,\nu)\models\psi$ (or just $(\mathcal A,<)\models\psi$ if $S=\emptyset$) if this holds when $\psi$ is seen as an \FO formula with free variables in $X_S$.

We say that valuation $\nu$ is \emph{consistent} if the $S$ is valid and if for every $(w,i)\in S$, $(\mathcal A,<,\nu)\models E(x_w^i,x_w^j)$ for some $j<i$.

\subsection{Cluster EF games.}
\label{section: CEF game}

As a tool to prove results about the expressive power of cluster first-order logic on ordered graphs, we introduce a variation on Ehrenfeucht--Fra\"{i}ss\'{e} games tailored for \CFO: the cluster EF game. In this game, Spoiler faces Duplicator: the former wants to show that two ordered graphs are different, while the latter tries to show their similarity with respect to \CFO. This two-player game provides a necessary and sufficient condition for the equivalence of two ordered graphs $\mathcal A$ and $\mathcal B$ when it comes to $\cfo$; more precisely, $\mathcal A$ and $\mathcal B$ are $k$-equivalent with respect to $\cfo$ if and only if the second player has a winning strategy for $k$ rounds of the game. This characterisation is a crucial tool for establishing inexpressibility results.

In order to state the winning condition of cluster EF game, we need the notion of partial isomorphisms. Given two structures $\mathcal A$ and $\mathcal B$ be structures over a relational vocabulary $\sigma$, a \emph{partial isomorphism} between $P \subseteq A$ and $Q \subseteq B$ is a bijection $f: P \to Q$ such that for every $n$-ary relation $R \in \sigma\cup\{=\}$ and every $a_1, \dots, a_n \in P$,
\[ (a_1, \dots, a_n) \in R^A \iff \big(f(a_1), \dots, f(a_n)\big) \in R^B\,. \]
In other words, it is a bijection witnessing that $\mathcal A_{|P}$ and $\mathcal B_{|Q}$ are isomorphic.

As a bit of notation, given a valuation $\nu:X\to A$, a variable $x\notin X$ and an element $a\in A$, we write $\nu+x\mapsto a$ to denote the valuation extending $\nu$ to $X\cup\{x\}$ by mapping $x$ to $a$.

Given an integer $k$, a finite set $S$ of indices, two ordered graphs $(\mathcal A,<_A), (\mathcal B,<_B)$ and two consistent valuations $\nu_A:X_S\to A, \nu_B:X_S\to B$, the \emph{$k$-round cluster EF game between $(\mathcal A,<_B,\nu_A)$ and $(\mathcal B,<_B,\nu_B)$} is a two-player game involving \emph{Spoiler} and \emph{Duplicator}. 

A \emph{configuration} of the game is such a tuple $\big(S,(\mathcal A,<_A,\nu_A),(\mathcal B,<_B,\nu_B)\big)$. For each $(w,i)\in S$ we introduce \emph{pebbles} $p_w^i$ and $q_w^i$, and say that $p_w^i$ (resp. $q_w^i$) is placed on $\nu_A(x_w^i)$ (resp. $\nu_B(x_w^i)$). A round of the game, starting from configuration $\big(S,(\mathcal A,<_A,\nu_A),(\mathcal B,<_B,\nu_B)\big)$, proceeds as follows to define the next configuration:

\begin{enumerate}
\item Spoiler picks a structure $(\mathcal A,<_A)$ and $(\mathcal B,<_B)$. In the following, we assume they chose $(\mathcal A,<_A)$ -- the case $(\mathcal B,<_B)$ is defined analogously.
\item Spoiler then makes either a(n)...
  \begin{description}
  \item[...introduction move (if $S=\emptyset$):] Spoiler places pebble $p_\epsilon^0$ on an element of $A$.
  \item[...introduction move (if $S\neq\emptyset$):] Spoiler chooses $w\in\Sigma^*$ and $\alpha\in\Sigma$ such that $(w,0)\in S$ and $(w\alpha,0)\notin S$. Then they place pebble $p_{w\alpha}^0$ on a element of $A$.
  \item[...continuation move:] Spoiler chooses $w\in\Sigma^*$ such that $(w,0)\in S$, and places $p_{w}^i$ (where $i=\min\{j:(w,j)\notin S\}$) on a element of $A$ that is edge-adjacent to some previous $p_w^j$, $j<i$.
  \end{description}
\item Then Duplicator answers by placing the corresponding pebble in the other structure. In the case of a(n)...
  \begin{description}
  \item[...introduction move (if $S=\emptyset$),] Duplicator places $q_\epsilon^0$ on an element of $B$.
  \item[...introduction move (if $S\neq\emptyset$),] Duplicator places $q_{w\alpha}^0$ on an element of $B$.
  \item[...continuation move,] Duplicator places $q_w^i$ on an element of $B$ that is adjacent to $q_w^j$. If such an element does not exist, then Duplicator loses and the game stops.
  \end{description}
\end{enumerate}

In the case of a(n)...
\begin{description}
\item[...introduction move (if $S=\emptyset$):] let $S':=\{(\epsilon,0)\}$, let $\nu_A':=x_\epsilon^0\mapsto p_\epsilon^0$, and let $\nu_B':=x_\epsilon^0\mapsto q_\epsilon^0$.
\item[...introduction move (if $S\neq\emptyset$):] let $S':=S\cup\{(w\alpha,0)\}$, let $\nu_A':=\nu_A+x_{w\alpha}^0\mapsto p_{w\alpha}^0$, and let $\nu_B':=\nu_B+x_{w\alpha}^0\mapsto q_{w\alpha}^0$.
\item[...continuation move:] let $S':=S\cup\{(w,i)\}$, let $\nu_A':=\nu_A+x_w^i\mapsto p_w^i$, and let $\nu_B':=\nu_B+x_w^i\mapsto q_w^i$.
\end{description}
Note that in each case, $S'$ is valid and $\nu_A',\nu_B'$ are consistent valuations. The configuration resulting from this round is $\big(S',(\mathcal A,<_A,\nu_A'),(\mathcal B,<_B,\nu_B')$.

A configuration $\big(S,(\mathcal A,<_A,\nu_A),(\mathcal B,<_B,\nu_B)\big)$ is \emph{winning for Duplicator} if:
\begin{itemize}
\item $\forall w\in\Sigma^*$ such that $(w,0)\in S$, $(p_w^i\mapsto q_w^i)_{i:(w,i)\in S}$ is a partial isomorphism for unordered structures $\mathcal A$ and $\mathcal B$, and
\item $\forall w\in\Sigma^*,\forall\alpha\in\Sigma$ such that $(w\alpha,0)\in S$,
  \[
  \begin{cases}
    p_w^i\mapsto q_w^i\quad\text{for $i$ such that }(w,i)\in S'\\
    p_{w\alpha}^0\mapsto q_{w\alpha}^0
  \end{cases}
  \]
  is a partial isomorphism when restricting the vocabulary to $\{<\}$. 
\end{itemize}
In words, Duplicator wins if each cluster induces isomorphic substructures on both sides (with indices taken into account), and if the first vertices of each non-$\epsilon$ cluster are in the same relative $<$-position with respect to their larger-proper-prefix cluster.

We say that Duplicator wins the $k$-round cluster EF game between $(\mathcal A,<_A,\nu_A)$ and $(\mathcal B,<_B,\nu_B)$ if, starting from configuration $\big(S,(\mathcal A,<_A,\nu_A),(\mathcal B,<_B,\nu_B)\big)$ (where $S$ is the set of indices such that $\nu_A$ and $\nu_B$ have domain $X_S$), they can play $k$ rounds and end up in a winning configuration no matter how Spoiler plays. If $S=\emptyset$, we simply say that \emph{Duplicator wins} (or \emph{has a winning strategy} for) the $k$-round cluster EF game between $(\mathcal A,<_A)$ and $(\mathcal B,<_B)$.

Unsurprisingly, the cluster EF game was tailored to capture exactly the expressive power of \CFO:

\CEFcapturesCFO*

We actually prove the following strengthening of Theorem~\ref{th:CEF_captures_CFO}, and get back the Theorem as the special case $S=\emptyset$:
\begin{lemma}
\label{proof of CEF capturing CFO}
  Let $(\mathcal A,<_A)$ and $(\mathcal B,<_B)$ be two ordered graphs, let $S$ be a valid set of indices and let $\nu_A:X_S\mapsto A$ and $\nu_A:X_S\mapsto A$ be two consistent valuations. Then $(\mathcal{A},<_A,\nu_A)\equiv^k_{\cfo}(\mathcal{B},<_B,\nu_B)$ if and only if Duplicator has a winning strategy in the $k$-round cluster EF game between $(\mathcal A,<_A,\nu_A)$ and $(\mathcal B,<_B,\nu_B)$.
\end{lemma}
Before presenting the proof, let us state the following lemma.
\begin{lemma}
\label{lemma: finitely many sentences}
If $\sigma$ is finite, then up to logical equivalence, $\cfo_S[k]$ over $\sigma\cup\{<\}$ contains only finitely many formulas.
\end{lemma}

\begin{proof}
\label{lemma saying that there are finitely many CFO with k quantifiers}
This follows from the inclusion of $\cfo[k]$ in $\FO[k]$ on $\sigma\cup\{<\}$, and from the fact that when the vocabulary is finite, $\FO[k]$ contains only finitely many formulas with free variables in the finite set $X_S$, up to logical equivalence (see \cite{DBLP:books/sp/Libkin04}, Chapter~3, Lemma~3.13).
\end{proof}

\begin{proof}[Proof of \Cref{proof of CEF capturing CFO}]
  Before proving this result, we need to remark that since
  \begin{itemize}
  \item the number non-equivalent formulas of $\CFO_S[k]$ is finite, and
  \item $\CFO_S[k]$ is closed under boolean combinations,
  \end{itemize}
  there exists, for each $(\mathcal A,<_A,\nu_A)$, a formula $\Psi_\CFO^k(\mathcal A,<_A,\nu_A)\in\CFO_S[k]$ that describes precisely the $\CFO_S[k]$-type of $(\mathcal A,<_A,\nu_A)$, in the following sense: any ordered graph and valuation $(\mathcal B,<_B,\nu_B)$ satisfies $\Psi_\CFO^k(\mathcal A,<_A,\nu_A)$ iff $(\mathcal B,<_B,\nu_B)\equiv_\CFO^k(\mathcal A,<_A,\nu_A)$.

  We can now show this lemma by induction on $k$. For $k=0$, note that congiguration $\big(S,(\mathcal A,<_A,\nu_A),(\mathcal B,<_B,\nu_B)\big)$ is winning for Duplicator exactly if $(\mathcal A,<_A)$ and $(\mathcal B,<_B)$ satisfy the same atomic formulas of $\CFO_S[0]$ with respective valuations $\nu_A$ and $\nu_B$ (and thus, all the formulas of $\CFO_S[0]$, which are just boolean combinations of those). In other words, Duplicator wins the $0$-round game iff $(\mathcal{A},<_A,\nu_A)\equiv^0_{\cfo}(\mathcal{B},<_B,\nu_B)$.

  Moving to the inductive case, let us assume that the lemma holds for some $k\geq 0$. We start by showing the left-to-right implication for $k+1$. To that end, let us assume that $(\mathcal{A},<_A,\nu_A)\equiv^{k+1}_{\cfo}(\mathcal{B},<_B,\nu_B)$, and let us try to ensure Duplicator wins the $(k+1)$-round game starting from configuration $\big(S,(\mathcal A,<_A,\nu_A),(\mathcal B,<_B,\nu_B)\big)$. Without loss of generality, let us assume Spoiler plays on $(\mathcal A,<_A)$. To that end, we have to consider the three kinds of moves and make sure Duplicator can answer so as to have a $k$-round winning strategy after this initial round. If Spoiler makes a(n)...
  \begin{description}
  \item[...introduction move (with $S=\emptyset$)] and places pebble $p_\epsilon^0$ in $A$, then we consider formula $\varphi:=\exists x_\epsilon^0\Psi_\CFO^k(\mathcal A,<_A,x_\epsilon^0\mapsto p_\epsilon^0)$.
    By construction, $(\mathcal A,<_A)\models\varphi$ and $\varphi\in\CFO_\emptyset[k+1]$. It follows by hypothesis that $(\mathcal B,<_B)\models\varphi$, and there exists an element of $B$, on which Duplicator will place pebble $q_\epsilon^0$, such that $(\mathcal B,<_B,x_\epsilon^0\mapsto q_\epsilon^0)\models\Psi_\CFO^k(\mathcal A,<_A,x_\epsilon^0\mapsto p_\epsilon^0)$. By definition, this entails $(\mathcal B,<_B,x_\epsilon^0\mapsto q_\epsilon^0)\equiv_\CFO^k(\mathcal A,<_A,x_\epsilon^0\mapsto p_\epsilon^0)$, ensuring by inductive hypothesis that Duplicator has a $k$-round winning strategy from configuration $\big(\{(\epsilon,0)\},(\mathcal A,<_A,x_\epsilon^0\mapsto p_\epsilon^0), (\mathcal B,<_B,x_\epsilon^0\mapsto q_\epsilon^0)\big)$.
  \item[...introduction move (with $S\neq\emptyset$)] and places pebble $p_{w\alpha}^0$ in $A$, then we proceed as above and consider formula $\varphi:=\exists x_{w\alpha}^0\Psi_\CFO^k(\mathcal A,<_A,\nu_A+x_{w\alpha}^0\mapsto p_{w\alpha}^0)$. We argue in the exact same way that Duplicator can place $q_{w\alpha}^0$ in $B$ in such a way as to have a $k$-round winning strategy from configuration $\big(S\cup\{(w\alpha,0)\},(\mathcal A,<_A,\nu_A+x_{w\alpha}^0\mapsto p_{w\alpha}^0), (\mathcal B,<_B,\nu_B+x_{w\alpha}^0\mapsto q_{w\alpha}^0)\big)$.
  \item[...continuation move] and places pebbles $p_w^i$ in $A$ on an element that is edge-adjacent to some $p_w^j$ (with $j<i$), then we consider formula $\varphi:=\exists x_w^iE(x_w^i,x_w^j)\land\Psi_\CFO^k(\mathcal A,<_A,\nu_A+x_w^i\mapsto p_w^i)$. Once again, $(\mathcal B,<_b,\nu_B)\models\varphi$, and there exists a element of $B$, on which Duplicator places $q_w^i$, such that $(\mathcal B,<_B,\nu_B+x_w^i\mapsto q_w^i)\models\Psi_\CFO^k(\mathcal A,<_A,x_w^i\mapsto p_w^i)$. Note that this is a legal answer for Duplicator, as $q_w^i$ is guaranteed by $\varphi$ to be edge-adjacent to $q_w^j$. Again, Duplicator has a $k$-round winning strategy from configuration $\big(S\cup\{(w,i)\},(\mathcal A,<_A,\nu_A+x_w^i\mapsto p_w^i), (\mathcal B,<_B,\nu_B+x_w^i\mapsto q_w^i)\big)$.
  \end{description}
  No matter the choice of Spoiler, Duplicator can answer and guarantee to be able to win for $k$ more rounds: this means Duplicator has a winning strategy in the initial $k+1$-round game, concluding this direction of the inductive case.

  It remains to prove the right-to-left implication of the inductive case; to that end, let us assume that Duplicator has a winning strategy in the $(k+1)$-round game starting in configuration $\big(S,(\mathcal A,<_A,\nu_A),(\mathcal B,<_B,\nu_B)\big)$, and let us show that $(\mathcal{A},<_A,\nu_A)$ and $(\mathcal{B},<_B,\nu_B)$ agree on all formulas of $\CFO_S^{k+1}$. Such a formula is a boolean combination of existential formulas of quantifier rank $k$, and we only need to show that $(\mathcal{A},<_A,\nu_A)$ and $(\mathcal{B},<_B,\nu_B)$ agree on these formulas. They existential formulas can be of three forms:
  \begin{itemize}
  \item $\varphi=\exists x_\epsilon^0\psi$ (only possible if $S=\emptyset$), where $\psi\in\CFO_{\{(\epsilon,0)\}}[k]$. If $(\mathcal A,<_A)\models\varphi$, then Spoiler can place $p_\epsilon^0$ in $A$ (an introduction move with $S=\emptyset$) so that $(\mathcal A,<_A,x_\epsilon^0\mapsto p_\epsilon^0)\models\psi$.  By hypothesis, we know that Duplicator can answer this introduction move by placing $q_w^0$ in $B$ so as to have a $k$-round winning strategy from configuration $\big(\{(\epsilon,0)\},(\mathcal A,<_A,x_\epsilon^0\mapsto p_\epsilon^0),(\mathcal B,<_B,x_\epsilon^0\mapsto q_\epsilon^0)\big)$. By induction hypothesis, this means $(\mathcal A,<_A,x_\epsilon^0\mapsto p_\epsilon^0)$ and $(\mathcal B,<_B,x_\epsilon^0\mapsto q_\epsilon^0)$ agree on $\CFO_{\{(\epsilon,0)\}}[k]$ and, in particular, on $\psi$. Thus $(\mathcal B,<_B)\models\varphi$. Conversely, if $(\mathcal B,<_B)\models\varphi$ then we make Spoiler play on $B$ and conclude that $(\mathcal A,<_A)\models\varphi$.
  \item $\varphi=\exists x_{w\alpha}^0\psi$ where $\psi\in\CFO_{S\cup\{(w\alpha,0)\}}[k]$. As in the case above, if $(\mathcal A,<_A,\nu_A)\models\varphi$ (the converse case being similar) then we make Spoiler play an introduction move (with $S\neq\emptyset$) on an element of $A$ witnessing $\psi$, and use Duplicator's answer in $B$ as a witness to $(\mathcal B,<_B,\nu_B)\models\varphi$.
  \item $\varphi=\exists x_w^iE(x_w^i,x_w^j)\land\psi$, where $j<i$ and $\psi\in\CFO_{S\cup \{(w,i)\}}[k]$. Suppose $(\mathcal A,<_A,\nu_A)\models\varphi$ (the case $(\mathcal B,<_B,\nu_B)\models\varphi$ is symmetric): then Spoiler can place $p_w^i$ in $A$ on a witness to $\varphi$, i.e. so that $(\mathcal A,<_A,\nu_A+x_w^i\mapsto p_w^i)\models E(x_w^i,x_w^j)\land\psi$. This is a valid move, as $p_w^i$ must then be edge-adjacent to $p_w^j$. By hypothesis, Duplicator knows how to answer by placing $q_w^i$ in $B$ on an element adjacent to $q_w^j$ in such a way to be able to win for $k$ more rounds from configuration $\big(S\cup\{(w,i)\},(\mathcal A,<_A,\nu_A+x_w^i\mapsto p_w^i),(\mathcal B,<_B,\nu_B+x_w^i\mapsto q_w^i)\big)$. By induction hypothesis, $(\mathcal A,<_A,\nu_A+x_w^i\mapsto p_w^i)$ and $(\mathcal B,<_B,\nu_B+x_w^i\mapsto q_w^i)$ agree on $\CFO_{S\cup \{(w,i)\}}[k]$, and a fortiori on $\psi$. All in all, we indeed conclude that $(\mathcal B,<_B,\nu_B)\models\varphi$.
  \end{itemize}
\end{proof}
\end{toappendix}


%% file: BoundedDegree.tex
\begin{toappendix}
\section{Expressive power}

In this section, we prove that order-invariant cluster first-order logic is included in first-order logic when the degree is bounded (Theorem~\ref{theorem: main result}). We proceed by introducing orders, called $(k,F)$-orders, such that extending two $\FO$-similar graphs with a suitable $(k,F)$-orders, keeps them $\cfo$-similar (Theorem~\ref{main technical theorem}). Therefore, a large part of this section is devoted to the presentation of these orders.

Before defining the $(k,F)$-order, we must establish the structural building blocks used to measure local similarity under an ordering. We achieve this by recursively defining the \emph{$k$-context} of an element. Unlike a standard neighbourhood, a context summarises the local environment of a vertex by recording its colour and the restricted structure of its $k$-neighbourhood, while simultaneously accounting for the relative positions of elements within the intervals created by the linear order. This recursive approach ensures that the context captures not just the presence of neighbours, but how $(k-1)$-contexts are distributed between them. This refinement is essential for the Duplicator's strategy, as sharing a $k$-context implies sharing all lower-level contexts, allowing for a consistent response across different ordered graphs.

The $(k,F)$-order is a linearisation of a bounded-degree graph. The construction partitions the domain into a sequence of segments that isolate \emph{rare} local configurations, those appearing fewer than a specific threshold, and provides standardised \emph{universal} reservoirs for frequent ones. By organising elements into an extremity of rare types, iterative universal segments that act as copies of realisable contexts, and a central \emph{jungle} for the remaining elements, the order creates a predictable environment. High-frequency contexts are buffered by multiple layers of neighbouring segments to ensure that an element's local neighbourhood remains contained within the ordered parts and does not spill into the jungle, thereby allowing the Duplicator to mirror moves across identical segments in different structures.

Given an ordered graph $(\struct,<)$, a subset $B$ of its domain and two elements $a<b$ of $B$ that are consecutive in $B$ for $<$ (that is, such that no element of $B$ is $<$-between them), we let the \emph{interval} $]a,b[$ be the tuple $(a,b,I)$ where $I$ is the set of elements of $A$ that are strictly contained between $a$ and $b$ in the sense of $<$. We will write ``$x\in]a,b[$'' to mean ``$x\in I$''. We denote by $\Int[B]$ the set of all the intervals of $B$ in $(\struct,<)$; note that $|\Int[B]|=|B|-1$.

\begin{definition}[Context]
\label{def:context}
Let $\sigma$ be a relational vocabulary. We define below the \emph{$k$-context} of an element $v$ in an ordered graph $(\struct,<)$ by induction on $k$. We denote by $\allcontext^k$ the set of all $k$-contexts:
$$\allcontext^{k}:=\{\mathds{C}^k_{(\struct,<)}(v)\text{ for every ordered graph } (\struct,<)\text{ and }v\in A\}.$$
\begin{enumerate}
	\item Let $\mathds{C}^{0}_{(\struct,<)}(v)$ be the atomic type of $v$ in $(\struct,<)$, \ie its colour.
	\item Let $(\struct,<)|_{N^k_{\struct}(v)}$ be the restriction of the ordered graph $(\struct,<)$ to the $k$-neighbourhood of $v$ in $\struct$, let
          \[f:
          \begin{cases}
            \Int[\boule[k][v]]&\rightarrow \mathcal{P}(\allcontext^{k-1})\\
  ]a,b[&\mapsto \{\context[k-1][x]:x\in ]a,b[\}
          \end{cases}
          \]
          and
          \[g:
          \begin{cases}
            \boule[k][v]&\rightarrow \allcontext^{k-1}\\
            a&\mapsto \context[k-1][a]
          \end{cases}\,.
          \]
          The $k$-context $\mathds{C}^k_{(\struct,<)}(v)$ of $v$ in $(\struct,<)$ is defined as
	$$\mathds{C}^k_{(\struct,<)}(v):=\big((\struct ,<)|_{N^k_{\struct}(v)},~g,~f\big).$$
\end{enumerate}

A $k$-context $c\in \allcontext^{k}$ is \emph{realisable} in a graph \struct if there exists an ordering $<$ and an element $v$ of $A$ such that  $c=\con^{k}_{(\struct,<)}(v)$. In this case we say that $v$ \emph{realises $c$ in $(\struct,<)$}. The set of \emph{elements of the context of $v$ in $(\struct ,<)$} is also defined recursively.
\begin{enumerate}
	\item $\econ^0_{(\struct ,<)}(v):=\{v\}$, and
	\item $\econ^{k}_{(\struct,<)}(v):=N^{k}_{\struct}(v)\cup \bigcup_{]a,b[\in \Int[\boule[k][v]]}\bigcup_{x \in ]a,b[} \econ^{k-1}_{(\struct,<)}(x)$.
\end{enumerate}
We say that a $k$-context $c$ is minimally realised if the size of its set of elements is minimal.
Two realisable $k$-contexts $c_1$ and $c_2$ are said to be disjoint if the sets of their elements are disjoint.

\end{definition}

As in this section our main focus is on the classes of graphs with bounded degree, say $d$, by $\allcontext^{k}_{d}$ we denote the set of all $k$-contexts that are realisable in structures with bounded degree $d$ and vocabulary $\sigma$.
By $\nc^k_d$, we denote the number of realisable $k$-contexts in structures with bounded degree $d$ and vocabulary $\sigma$, or equivalently  $\nc^{k}_{d}:=|\allcontext^{k}_{d}|$, and by $\bc^k_d$, we denote the bound on the size of a minimally realisable $k$-context in structures of bounded degree $d$.

\begin{lemma}
	\label{lemma: if k context is the same k-1 context is the same}
	Let $(\mathcal A,<_A)$ and $(\mathcal B,<_B)$ be two ordered graphs. For every $a\in A$, $b\in B$ and integer $k'<k$, $\mathds{C}^k_{(\mathcal A,<_A)}(a)=\mathds{C}^k_{(\mathcal B,<_B)}(b)$ entails $\mathds{C}^{k'}_{(\mathcal A,<_A)}(a)=\mathds{C}^{k'}_{(\mathcal B,<_B)}(b)$.
\end{lemma}
\begin{proof}
We show the lemma for $k'=k-1$ by induction on $k$. The general result for $k'<k$ follows immediately. For the base case $k=1$, it is clear as per the definition we remember the $0$-context of the elements as well. For the induction hypothesis assume that for any two elements if $\mathds{C}^{l}_{(A,<_A)}(a)=\mathds{C}^{l}_{(B,<_B)}(b)$, then $\mathds{C}^{l-1}_{(A,<_A)}(a)=\mathds{C}^{l-1}_{(B,<_B)}(b)$, for all integers $0<l<k$. Now, if $\mathds{C}^{k}_{(A,<_A)}(a)=\mathds{C}^{k}_{(B,<_B)}(b)$, then $\mathds{C}^0_A(a)=\mathds{C}^0_{B}(b)$, again by the force of the definition. The equality of the $k$-contexts imply that the corresponding elements of the $k$-neighbourhood respecting the orders have the same $k-1$-contexts. Hence, as the restriction of the orders to $k-1$-neighbourhood instead of the $k$-neighbourhood is just a subsequence of the same order, we have the same corresponding elements with the same $k-1$-contexts. By induction hypothesis, they have the same $k-2$-contexts. Therefore, looking at $(A,<_A)|_{N^{k-1}_A(a)},g^{k-1}_a$ and $(B,<_B)|_{N^{k-1}_B(b)},g^{k-1}_b$ we see the same order on elements with the same $k-2$-contexts.

For the fourth attribute, we have that in any interval of the $k-1$-neighbourhood, we have the elements at distance exactly $k$, or the same elements that existed in an interval of the $k$-neighbourhood. For both cases, we know by the assumption that they admit the same set of $k-1$-contexts and hence by induction hypothesis the same set of $k-2$-contexts.
\end{proof}

In Definition~\ref{def:context}, function $f$ associates each interval of \Int[\boule[k][v]] with the set of $(k-1)$-contexts that are realised in it.
There, an interval is always included between two elements of \boule[k][v] that are consecutive for $<$. It will be convenient to extend this notion so that it takes into account two new intervals: the leftmost and the rightmost.

\begin{definition}[Outer context]
  \label{def:outer_context}
  We extend the set \Int[B] of intervals of $B$ in $(A,<)$ to \OInt[B] by adding two new intervals $]-\infty,\min_<(B)[$ and $]\max_<(B),\infty[$, where $\min_<(B)$ and $\max_<(B)$ are the $<$-minimal and $<$-maximal elements of $B$. These intervals contain respectively all the elements of $A$ which are strictly smaller that $\min_<(B)$, or larger that $\max_<(B)$.
      
      The \emph{outer $k$-context} \ocontext[k][v] of an element $v$ in $(\struct,<)$ is then defined recursively in the same way as \context[k][v], except that the domain of $f$ is now \OInt[\boule[k][v]] instead of \Int[\boule[k][v]].
\end{definition}

\begin{definition}[Neighbourhood type]
	Let $c$ be a constant symbol, and let $\mathcal{A}\in\mathcal{C}$ be a structure in a class $\mathcal{C}$ of structures. We define, for $k\in\N$ and $a\in A$, the \emph{pointed $k$-neighbourhood $\ntype{k}{A}{a}$} to be the restriction of the structure $\mathcal{A}$ with vocabulary $\sigma\cup\{c\}$ to the $k$-neighbourhood of $a$, i.e., $N^k_{A}(a):=\{b\in A\mid \dist_{A}(a,b)\leq k\}$, where $c$ is interpreted as $a$. To put it differently, the structure $\ntype{k}{A}{a}$ is the substructure induced by the $k$-neighbourhood of $a$ in $A$, where we label $a$ with $c$.

By \emph{$k$-neighbourhood types of $\mathcal{A}$} we mean the \emph{isomorphism class of pointed $k$-neighbourhoods}. If for an element $a\in A$, the pointed $k$-neighbourhood $\ntype{k}{A}{a}$ belongs to the isomorphism class $\tau$, we say that $a$ is an \emph{occurrence} of $\tau$ in $A$. By $|\mathcal{A}|_\tau$ we denote the number of occurrences of $\tau$ in $\mathcal{A}$, i.e., the number of elements $a\in A$ with pointed $k$-neighbourhood $\tau$.

Let $\mathcal{A}_1$ and $\mathcal{A}_2$ be two structures, we write $\llbracket \mathcal{A}_1\rrbracket_k=^{t}\llbracket \mathcal{A}_2\rrbracket_k$ whenever for any $k$-neighbourhood type $\tau$, either $|\mathcal{A}_1|_\tau =|\mathcal{A}_2|_\tau$ or both values are greater than $t$. In this scenario we say that $\mathcal{A}_1$ and $\mathcal{A}_2$ are \emph{$(k,t)$-threshold equivalent}.
By $\neightype^d_k(\mathcal{C})$ we denote the set of different \emph{$k$-neighbourhood types} occurring in such a class $\mathcal{C}$. When the class $\mathcal{C}$ is evident in the context we avoid writing it, instead we use $\neightype^d_k$. 
\end{definition}

\begin{observation}
\label{observation: neighbourhood types are finite}
Let $\mathcal{C}$ be a class of structures with bounded degree $d$, then the set of $k$-neighbour\-hood types occurring in these structures is finite. 
\end{observation}

\begin{proof}[Proof sketch]
For fixed $k$ and $d$, every $k$-neighbourhood has size bounded by a function of $k$ and $d$. Since the vocabulary $\sigma$ is finite, there are only finitely many structures of this bounded size over $\sigma\cup\{c\}$, up to isomorphism. Hence only finitely many pointed $k$-neighbourhood types can occur.
\end{proof}

\begin{corollary}
\label{corollary: all contexts in finite}
		Let $k$ and $d$ be integers, then $|\allcontext^{k}_{d}|$ and the size of a minimally realised $k$-context are finite and only depends on $k,d$ and $|\sigma|$.
\end{corollary}

\begin{proof}
\label{proof of all contexts in finite}
Proof by induction on $k$. For $k=0$, we have that $\allcontext^{0}_{d}$ is exactly the set of colours therefore $|\allcontext^{0}_{d}|=|\sigma|$. A $k$-context $c\in \allcontext^{k}_{d}$ can be constructed by a vertex $v$ with the same $k$-neighbourhood types as a vertex that realises $c$ in a structure, for which there are $
|\neightype^{k}_{d}|$ possibilities. This neighbourhood is of size at most $d^{k}-1$, therefore there are $(d^{k}-1)!$ ordering on this neighbourhood. Between each consecutive element of this neighbourhood we can have up to $\nc^{k-1}_{d}$ many different $k-1$-contexts. Hence, there are at most
$$\nc^{k}_{d}=|\neightype^{k}_{d}|\cdot(d^{k}-1)!\cdot (2^{\nc^{k-1}_{d}})^{(d^{k}-2)}$$
different $k$-contexts. By Observation~\ref{observation: neighbourhood types are finite}, $|\neightype^{k}_{d}|$ is bounded; hence $\nc^{k}_{d}$ is bounded. Now, the size of a minimally realised $0$-context is $1$, \ie $\bc^0_d=1$. 
In a $k$-context, the vertex that realises this context in a structure has a neighbourhood of size $d^{k}-1$, and there are $d^{k}-2$ gaps between the consecutive elements of this neighbourhood, in each of which we can have at most $\nc^{k-1}_{d}$ many different $k-1$-contexts of size $\bc^{k-1}_{d}$. Therefore, the size of a minimally realised $k$-context, \ie $\bc^k_d$, is bounded by $$d^{k}-1+ (d^{k}-2)\cdot \nc^{k-1}_{d}\cdot \bc^{k-1}_{d}.$$
\end{proof}

When speaking about realisability of a context $c$, it is usually in a structure and it means that there is an order $<$ such that when ordering the said structure with $<$, one can find a vertex in the ordered structure that has the context $c$ (see \Cref{def:context}). However, given a $k$-context $c$ in a structure $(\mathcal A,<)$, one can throw away all the elements of the structure that are not included in the set of elements of $c$ in $(\mathcal A,<)$. The remaining elements, \ie the set of elements of $c$ in $(\mathcal A,<)$, consist of some vertices and their $k$-neighbourhoods that are arranged in a specific way, forming the context $c$. In this way, we can comprehend a $k$-context outside of the scope of a structure, as an ordering on a set of vertices together with their neighbourhoods.

\begin{definition}
Let $i,k$ be integers such that $i\leq k$, and $F$ and $c$ be a set of $k$-neighbourhood types and an $i$-context, respectively. We say that $c$ is realisable with the $k$-neighbourhood types of $F$ if there exists an ordered structure $(\mathcal A,<)$ and a vertex $v\in A$, such that $\mathds{C}^i_{(\mathcal A,<)}(v)=c$	and the set of $k$-neighbourhood types occurring in the unordered structure $\mathcal A$, for the elements of the set of elements of $c$ in $(\mathcal A,<)$, \ie, $\mathds{E}^i_{(\mathcal A,<)}(v)$, is included in $F$. More formally
\[
\left\{ \tau \;\middle|\; a \in \mathds{E}^i_{(\mathcal A,<)}(v)
\text{ such that } \tau \text{ is the } k\text{-neighbourhood type of } a \text{ in } \mathcal A \right\}
\subseteq F.
\]
\end{definition}


\begin{definition}[$(m,r,s)$-frequent types]
	Let $k\in\N$. We say that a $k$-neighbourhood type $\tau$ is $(m,r,s)$-frequent in $\mathcal{A}$ if for any set $B\subseteq A$ where $|B|\leq s$ there are at least $m$ occurrences of $\tau$ where they are mutually at distance at least $r$ from each other and from $B$, i.e., we can find $a_1,\cdots,a_m$ where they all have $k$-neighbourhood type $\tau$ and $\dist(a_i,a_j)> r$, for $1\leq i,j\leq m$ and $\dist(a_i,B)> r$, for $1\leq i\leq m$.
\end{definition}

In \cite{bednarczyk2025expressive}, Bednarczyk and Grange showed that for any given $k,m$ and $r$ there exist an $s$ and a threshold $t$ such that we can partition the $k$-neighbourhood types into $(m,r,s)$-frequent types and the rest of the types, which we call then rare types. Note that elements with $(m,r,s)$-frequent types definitely appear more than $m$ times in the structure; moreover for any fixed set $B$ of size $s$ we can find $m$ copy of them that are mutually at distance greater than $r$, and each of them is at distance greater than $r$ from $B$ as well. They showed this through the following two Lemmas which we restate here without their proofs.
\begin{lemma}[Lemma 5.4 of \cite{bednarczyk2025expressive}]
\label{lemma: lemma 5.4}
Let $\mathcal{C}$ be a class of structures with bounded degree $d$ and $k$ be an integer. Then, there exists a function $g:\N\times\N\times\N\rightarrow\N$ such that for all $A\in\mathcal{C}$ and for any three positive integers $m,r$ and $s$, and any set $B\subseteq A$, with $|B|\leq s$,	and all subsets $C_1,\cdots,C_n\subseteq A$ ($n\leq |\neightype^d_k(\mathcal{C})|$) of size at least $g(m,r,s)$, we can find elements $c^1_j,\cdots,c^m_j\in C_j$ for all $1\leq j\leq n$, such that for all $1\leq j,j'\leq n$ and $1\leq i,i'\leq m$, $\dist_A(c^i_j,B)> r$ and $\dist_A(c^i_j,c^{i'}_{j'})>r$ if $(j,i)\neq (j',i')$.
\end{lemma}

For an integer $t$ and a structure $A\in\mathcal{C}$ let $\freq[A]^{\geq t}_k\subseteq \neightype^d_k(\mathcal{C})$ denote the set of $k$-neighbourhood types which have at least $t$ occurrences in $A$.
\begin{lemma}[Lemma 5.6 of \cite{bednarczyk2025expressive}]
\label{lemma: lemma 5.6}
	Let $\mathcal{C}$ be a class of structures with bounded degree $d$ and $k,m,r$ be integers. Then, there exists $\Upsilon\in \N$ such that for every $A\in\mathcal{C}$, there exists $t\leq \Upsilon$, for which
	$$t\geq g\left(m,r,\sum_{\tau\in \neightype^d_k(\mathcal{C})\setminus \freq[A]^{\geq t}_k}|A|_{\tau}\right),$$
	for the function $g$ from \Cref{lemma: lemma 5.4}.
\end{lemma}

The proof of Lemma~\ref{lemma: lemma 5.6} in \cite{bednarczyk2025expressive} is given for two specific values of $m$ and $r$, corresponding to those used subsequently in their paper. Nevertheless, a detailed examination of the argument reveals that the result holds in the more general setting stated above.

Now, we are ready to introduce the order.
\input{order.tex}

\begin{lemma}
 \label{main technical theorem}

  Let $<_A$ and $<_B$ be two $(F,k)$-orders of some graphs $\mathcal A$ and $\mathcal B$ of degree at most $d$, for a set $F$ that is a set of $k$-neighbourhood types. If a border-preserving bijection $\phi:A\setminus J^A\rightarrow B\setminus J^B$ exists, 
  then $(\mathcal A,<_A)\equiv_{\CFO}^k(\mathcal B,<_B)$.
\end{lemma}

\begin{proof}

To show $(A,<_A)\equiv^k_{\cfo}(B,<_B)$ we require a winning strategy for Duplicator. By induction on the number of rounds $r$, we describe a strategy satisfying the following properties after $r$ rounds. Let $\mathcal{I}$ be the indexing set of the game.
	\begin{itemize}
	\item[$\mathsf{S}_r$:] \emph{If a pebble has been played near the minimum or maximum, then it has been answered similarly}. More precisely, $$\forall l\leq r, p(l)\in S^{k-l}_A \iff q(l)\in S^{k-l}_B,$$ and in this case $q(l)=\phi(p(l))$.

	\item[$\mathsf{C}_r$:] \emph{Corresponding pebbles played at round $j$ are placed on elements with the same $(k-j)$-context.} Say pebbles $p(1),\cdots,p(r)$ and $q(1),\cdots,q(r)$ are played on $A$ and $B$, respectively. Then $\mathds{C}^{k-j}_{(\mathcal A,<_A)}(p(j))= \mathds{C}^{k-j}_{(\mathcal B,<_B)}(q(j))$ for $1\leq j\leq r$. Recall that in $\cfo$ indexing matters and each pebble is played with specific indices. For example $p(i)$ is of form $p^{j}_w$, for an integer $j$ and a word $w\in\Sigma^{*}$, and so $q(i)$ should be indexed $q^j_w$. 
          
        \item[$\mathsf{I}_r$:]  For every $w\in \mathcal I$, let $l\leq r$ be such that $p(l)=p_w^0$. Let

          \begin{itemize}
          \item $\mathcal N^+(p(l))$ be the structure $(\mathcal A,<_A)|_{\boule[k-l][p(l)][(\struct,<_A)]}$ enriched with a constant symbol $c_w^i$ interpreted as $p_w^i$ for each $i\in\N$ that as been played so far (note that $\boule[k-l][p(l)]$ must contain all such $p_w^i$ so $\mathcal N^+(p(l))$ is well-defined), and
          \item $h(p(l))$ be the function associating
            \begin{itemize}
            \item with each interval $I\in\OInt[\boule[k-l][p(l)]][(\struct,<_A)]$ the set of $\alpha\in\Sigma$ such that $w\alpha\in\mathcal I$ and $p_{w\alpha}^0\in I$, and
            \item with each element $a$ of $\neigh[k-l][p(l)][(\struct,<_A)]$ the set of $\alpha\in\Sigma$ such that $w\alpha\in\mathcal I$ and $p_{w\alpha}^0=a$,
            \end{itemize}
          \item and define similarly $\mathcal N^+(q(l))$ and $h(q(l))$.
          \end{itemize}

          Then $\mathcal N^+(p(l))$ and $\mathcal N^+(q(l))$ are isomorphic, and $h(p(l))$ and $h(q(l))$ coincide.
        \end{itemize}
	  Note that in round $k$, condition $\mathsf{I}_k$ is enough for Duplicator to win and hence what we are asking is more than the winning conditions.

	\textbf{\textit{Base case:} }
	Assume that Spoiler selects $\mathcal A$ to play the first pebble $p^0_\epsilon$ on; this assumption is for notational simplicity and the other case is symmetric.

        If Spoiler places $p^0_\epsilon$ on an element of $S^k_A$, the Duplicator replies by placing $q^0_\epsilon$ on $\phi(p^0_\epsilon)$, meaning that it plays a tit-for-tat move. As $A\setminus J^A$ and $B\setminus J^B$ are isomorphic under $\phi$, and \textbf{contraction} of the orders, corresponding elements of $S^k_A$ and $S^k_B$ have the same $k$-context.
        
        Otherwise, if Spoiler places $p^0_\epsilon$ on an element of $A\setminus S^k_A$, i.e., on the union of the jungle and the $2k^2$ neighbouring segments adjacent to the jungle, equivalently $J^A\cup\bigcup_{k^2+1}^{2k^2}LN^A_i\cup RN^A_i$, then Duplicator replies by placing $q^0_\epsilon$ on an element of $LU^B_{(k-1)^2}$ whose $(k-1)$-context is the same as that of $p^0_\epsilon$. By \textbf{universality}, we know that such an element exists. 
	
	Clearly $\mathsf{S}_1$ holds, $\mathsf{C}_1$ holds by choice of $q^0_\epsilon$, and the equality of the $k$-contexts of $p^0_\epsilon$ and $q^0_\epsilon$ entails $\mathsf{I}_1$.
	
	\textbf{\textit{Induction hypothesis:}}  The hypothesis is that $\mathsf{S}_r,\mathsf{C}_r$ and $\mathsf{I}_r$ hold.

	\textbf{\textit{Induction step:}}  Assume pebbles $p(1),\cdots,p(r)$ and $q(1),\cdots,q(r)$ have been played on $(\mathcal A,<_A)$ and $(\mathcal B,<_B)$, respectively. Again, we assume that Spoiler picks $\mathcal A$ as the structure it will play on, and the other case is symmetric. The proof is by a case study that has two major cases: one for introduction moves and one for continuation moves.

        \paragraph*{Introduction move.} Assume $p(r+1)=p_{w\alpha}^0$, and let $l\leq r$ be such that $p_w^0=p(l)$. We distinguish between two cases:
        \begin{itemize}
        \item Let $p^0_{w\alpha} \in S^{k-r-1}_A$. Duplicator responds with a tit-for-tat move by placing $q^0_{w\alpha}$ at $\phi(p^0_{w\alpha})$. Condition $\mathsf{S}_{r+1}$ is satisfied since a move in the safety part is answered by a move in the safety part. Condition $\mathsf{C}_{r+1}$ also holds, as the $k$-contexts of elements in the safety parts are entirely contained in $A \setminus J^A$ and $B \setminus J^B$, respectively. Consequently, the image of an element in the safety part of $A$, such as $p^0_{w\alpha}$, has the same $k$-context as $p^0_{w\alpha}$ itself. In particular, they share the same $(k-r-1)$-context.
       If $p^0_{w\alpha}$ is either in an interval $I\in\Int[\boule[k-l][p_w^0]][(\mathcal A,<_A)]$, or is an element of \boule[k-l][p_w^0], we can infer than $p^0_w\in S^{k-l}_A$, because the segment distance between $p^0_w$ and $p^0_{w\alpha}$ is then less than $2(k-l)$ which is less then $2k(r-l)$. Therefore, by $\mathsf{S}_l$, $q^0_w$ belongs to $S^{k-l}_B$ as well. The condition $\mathsf{I}_{r+1}$ is implicated as by playing according to $\phi$, on the contexts of $p^0_w$ and $q^0_w$, Duplicator ensures isomorphism on the substructures induced by the set of elements of $p^0_w$ and $q^0_w$. On the other hand, if $p^0_{w\alpha} $ does not belong to $I\in\Int[\boule[k-l][p_w^0]][(\mathcal A,<_A)]$, or \boule[k-l][p_w^0], with the same segment distance argument we infer that $p^0_w\notin S^{k-l}_A$ and therefore by $\mathsf{S}_{l}$ neither does $q^0_w$ belong to $S^{k-l}_B$. Accordingly, the $\mathsf{I}_{r+1}$ holds, as both $p^0_{w\alpha}$ and $q^0_{w\alpha}:=\phi(p^0_{w\alpha})$ are in the same relative position (right or left) with respect to the substructures induced by the set of elements of $p^0_w$ and $q^0_w$.
       
        \item Let $p^0_{w\alpha}\notin S^{k-r-1}_A$, we break this case into the following three subcases:
        \begin{itemize}
        	\item If $p^0_{w\alpha}$ is either in an interval $I\in\Int[\boule[k-l][p_w^0]][(\mathcal A,<_A)]$, or is an element of \boule[k-l][p_w^0], Duplicator places $q^0_{w\alpha}$ on the corresponding element of the corresponding interval of \Int[\boule[k-l][q_w^0][\mathcal B]][(\mathcal B,<_B)], or the corresponding element of \boule[k-l][q_w^0][\mathcal B] under $h$ given by $\mathsf{I}_{l}$. By this choice, $\mathsf{I}_{r+1}$ holds, by the definition of context, $p^0_{w\alpha}$ and $q^0_{w\alpha}$ have the same $k-l-1$-context and hence the same $k-r-1$-context. Therefore, $\mathsf{C}_{r+1}$ holds too. To show that $\mathsf{S}_{r+1}$ holds we need to argue that $q^0_{w\alpha}$ does not belong to $S^{k-r-1}_B$. Assume that $q^0_{w\alpha}\in S^{k-r-1}_B$. As $q^0_{w\alpha}$ is in an interval of \Int[\boule[k-l][q_w^0][\mathcal B]][(\mathcal B,<_B)], or is an element of \boule[k-l][q_w^0][\mathcal B], its segment distance from $q^0_w$ is less than $2(k-l)$, this forces $q^0_w$ to be in $S^{k-r}_B$ and then by $\mathsf{S}_{l}$, not only $p^0_w$ also belongs to $S^{k-r}_A$, but also they are each others image under $\phi$ and $\phi^{-1}$. Therefore, $p^0_{w\alpha}$ will have the same position as $q^0_{w\alpha}$ and so belongs to $S^{k-r-1}_A$ which is contradictory by the assumption we began with. Therefore, $\mathsf{S}_{r+1}$ holds too.
        	\item Assume now $p_{w\alpha}^0$ is in the leftmost interval of \OInt[\boule[k-l][p_w^0]][(\mathcal A,<_A)], that is $]-\infty,a[$ for $a=\min_{<_A}\boule[k-l][p_w^0]$, and that $p^0_{w\alpha}$ is in the left part of $S^{k-l}_A$ (i.e., $X^A\cup \bigcup_{i=1}^{(k-l)k} L^A_i$). This means that $p^0_{w\alpha}\in A\setminus J^A$. 
         Duplicator plays by placing $q^0_{w\alpha}$ on an element with the same $(k-r-1)$-context as $p^0_{w\alpha}$ in the leftmost universal segment that is adjacent to and comes after $S^{k-r-1}_B$, i.e., $LU^A_{k(k-r-2)+1}$. This is possible by the \textbf{universality} of the orders. By the way $q^0_{w\alpha}$ is chosen, $\mathsf{C}_{r+1}$ holds. As $q^0_{w\alpha}$ is outside of $S^{k-r-1}_B$(it is picked in the universal segment just after $S^{k-r}_B$), $\mathsf{S}_{r+1}$ holds too.
         
         As $p^0_w\in S^{k-l}_A$ by $\mathsf{S}_{l}$, we have that $q^0_{w}\in S^{k-l}_B$ 
         and so $q^0_{w\alpha}$ will be to the left of $q^0_{w}$. Assume that $q^0_{w\alpha}$ is not in the leftmost interval of \OInt[\boule[k-l][q_w^0][\mathcal B]][\mathcal B], which means it is in either an interval of \Int[\boule[k-l][q_w^0][\mathcal B]][\mathcal B] or is an element of \boule[k-l][q_w^0][\mathcal B]. As the segment distance between $q^0_{w\alpha}$ and $q^0_{w}$ will be at most $2(k-l)$ is this scenario, we can conclude that $q^0_{w}\in S^{k-r}_B$ and then by $\mathsf{S}_{r}$, we infer $p^0_w\in S^{k-r}_A$. This forces $p^0_{w\alpha}$ that is to the left of $p^0_w$ and in the left most interval of $\OInt[\boule[k-l][p_w^0]][(\mathcal A,<_A)]$ to be in $S^{k-r-1}_A$ which is contradictory with the assumption. Hence, $q^0_{w\alpha}$ is in the left most interval of  \OInt[\boule[k-l][q_w^0][\mathcal B]][\mathcal B] and $\mathsf{I}_{r+1}$ holds too.
         
        	\item Assume now $p_{w\alpha}^0$ is in the leftmost interval of \OInt[\boule[k-l][p_w^0]][(\mathcal A,<_A)], that is $]-\infty,a[$ for $a=\min_{<_A}\boule[k-l][p_w^0]$, and that $p^0_{w\alpha}$ is in the right part of $S^{k-l}_A$ (i.e., $X^A\cup \bigcup_{i=1}^{(k-l)k} R^A_i$) or does not belong to $S^{k-l}_A$. Duplicator plays by placing $q^0_{w\alpha}$ on an element with the same $(k-r-1)$-context as $p^0_{w\alpha}$ in the leftmost universal segment that is adjacent to and comes after $S^{k-r-1}_B$, i.e., $LU^A_{k(k-r-2)+1}$. This is possible by the \textbf{universality} of the orders. The conditions $\mathsf{C}_{r+1}$ and $\mathsf{S}_{r+1}$ hold with the exact same arguments as the previous case. 
        
        As $p^0_w\in S^{k-l}_A$ by $\mathsf{S}_l$, not only $q^0_w$  belongs to $S^{k-l}_B$ but also they are each others image under $\phi$ and $\phi^{-1}$. Therefore, $q^0{w}$ is also in the right part of $S^{k-l}_B$. Accordingly, $q^0_{w\alpha}$ will be in the leftmost interval of \OInt[\boule[k-l][q_w^0][\mathcal B]][(\mathcal B,<_B)]. Hence, $\mathsf{I}_{r+1}$ holds too.
        \end{itemize}
        The cases where $p^0_{w\alpha}$ is on the rightmost interval of \OInt[\boule[k-l][p_w^0]][(\mathcal A,<_A)] are symmetric to the last two subcases we just studied.

        \end{itemize}
        \paragraph*{Continuation move.}
        Assume $p(r+1)=p_{w}^i$ with $i>0$, and let $l\leq r$ be such that $p_w^0=p(l)$. By $\mathsf{I}_r$, the two structures $\mathcal N^+(p(l))$ and $\mathcal N^+(q(l))$ are isomorphic, and the adjacency constraint of continuation moves in cluster EF games enforce $p_w^i\in \mathcal N^+(p(l))$. Duplicator places $q_w^i$ on the image of $p_w^i$ under this isomorphism. This obviously maintains invariant $\mathsf{I}_{r+1}$. To show that $\mathsf{C}_{r+1}$ holds, note that $\boule[k-(r+1)][p(r+1)]$ and $\boule[k-(r+1)][q(r+1)]$ induce isomorphic subgraphs of $(\mathcal A,<_A)|_{\boule[k-l][p(l)]}$ and $(\mathcal B,<_B)|_{\boule[k-l][q(l)][\mathcal B]}$, which are themselves isomorphic in view of $\mathsf{C}_r$. It is an easy verification to check that functions $f$ and $g$ in $\mathds C^{k-(r+1)}_{(\mathcal A,<_A)}(p(r+1))$ and $\mathds C^{k-(r+1)}_{(\mathcal B,<_B)}(q(r+1))$ coincide (as a consequence of $f$ and $g$ coinciding between $\mathds C^{k-l}_{(\mathcal A,<_A)}(p(l))$ and $\mathds C^{k-l}_{(\mathcal B,<_B)}(q(l))$), thus equating these contexts. Finally, we show that $\mathsf{S}_{r+1}$ holds. Assume $p(r+1)\in S^{k-(r+1)}_A$ -- the other direction is symmetrical. The \textbf{Contraction} property ensures that $p(l)$, which is in the $k$-neighbourhood of $p(r+1)$, belongs to the $2k$-segment neighbourhood of $p(r+1)$. In particular, $p(l)\in S^{k-r}_A\subseteq S^{k-l}_A$. By $\mathsf{S}_r$, we get that $q(l)\in S^{k-l}_B$ and $q(l)=\phi(p(l))$. Note that $\phi$ is the only isomorphism between the linearly ordered structures $\mathcal N^+(p(l))$ and $\mathcal N^+(q(l))$, thus $q(r+1)=\phi(p(r+1))$, and we conclude by property of $\phi$ being a border-preserving bijection.
        
        \medskip
		
	After $k$ rounds, $\mathsf{I}_{k}$ guarantees Duplicator's winning conditions, and thus $(A,<_A)\equiv^k_{\cfo}(B,<_B)$.
\end{proof}

	We are ready the state the main result of the section, but before that let us state the  following theorem.
\begin{theorem}[Fagin, Stockmeyer, Vardi \cite{fagin1995monadic}]
\label{theorem: fagin, stockmeyer, vardi}
	Let $\mathcal{C}$ be a class of structures with bounded degree $d$, then there exists a function $f_{\mathsf{FSV}}:\N\times\N\rightarrow \N$ such that for any two structures $A_1,A_2\in\mathcal{C}$, if $A_1\equiv_{\FO}^{f_{\mathsf{FSV}}(k,t)} A_2$ then $\llbracket A_1\rrbracket _k=^t \llbracket A_2\rrbracket_k$.
\end{theorem}
Theorem \ref{theorem: fagin, stockmeyer, vardi} asserts that when two structures of bounded degree are sufficiently similar -- namely, when they are $\FO$-equivalent up to a sufficiently large quantifier rank $f_{\mathsf{FSV}}(k,t)$ -- then the counts of their $k$-neighbourhood types are equal up to threshold $t$.

	\begin{theorem}[Main theorem]
	\label{theorem: main result}
		 Let $\CCC$ be a class of structures with bounded degree $d$, for an integer $d$. Then, $<$-invariant $\cfo\subseteq \FO$ on $\CCC$.

\end{theorem}
\begin{proof}
To show the inclusion $<$-invariant $\cfo\subseteq \FO$ on $\CCC$, it is sufficient to show that there is a function $f:\N\rightarrow\N$ such that
 for any two structures $A,B\in\CCC$, if $A\equiv^{f(k)}_\FO B$, then $A\equiv^k_{<\textit{-invariant }\cfo}B$. In other words, we need to show that $\equiv^{f(k)}_\FO$ is a more fine-grained equivalence relation than $\equiv^k_{<\textit{-invariant }\cfo}$. This is because each of these equivalence relations partition the class $\CCC$ into equivalence classes. A formula $\varphi\in <-\textit{invariant }\cfo$, captures some of these equivalence classes, i.e., it is modelled by the structures in the union of some these classes. Now, it is folklore that there are only finitely many equivalence classes for $\equiv^{f(k)}_{\FO}$, each of which is definable by an $\FO$-sentence $\psi$. Therefore, if $c_1,\dots,c_l$ are the equivalence classes imposed by $\equiv^{f(k)}_\FO$ such that their union is equal to the union of the equivalence classes captured by $\varphi$, and $\psi_{c_1},\cdots,\psi_{c_l}$ are the $\FO$-sentences defining $c_1,\dots,c_l$, then $\psi_{c_1}\vee\cdots\vee\psi_{c_l}$ is modelled by the same structures modelling $\varphi$. Hence the inclusion.

 As an $<$-invariant $\cfo$ formula is invariant with respect to the interpretation of the order, showing $(A,<)\equiv^k_{\cfo}(B,<')$ for any orders $<$ and $<'$ would entail $A\equiv^k_{<\textit{-invariant }\cfo}B$. Thus, once again, the problem is reduced to showing that if $A\equiv^{f(k)}_\FO B$ then there are two orders $<_A$ and $<_B$ such that $(A,<_A)\equiv^k_{\cfo}(B,<_B)$.
 
Now, let $t_{k,d}$ and $\bar J_{k,d}$ be the integers given by \cref{good graphs admit good orders} for integers $k$ and $d$. We define 
 $f:\N\rightarrow\N$ as 
 $$f(k):= \max\{f_{\mathsf{FSV}}(k,t_{k,d}), \bar J_{k,d}\}\,.$$
 
 As $f(k)$ is greater than $f_{\mathsf{FSV}}(k,t_{k,d})$ we conclude that $\mathcal{A}$ and $\mathcal B$ are $(k,t_{k,d})$-threshold equivalent. Hence $\freq[\mathcal A]^{t_{k,d}}_k$ and $\freq[\mathcal{B}]^{t_{k,d}}_k$ are equal. 
 
 First, noice that by \cref{good graphs admit good orders}, the graph $\mathcal A$ admits an $(k,F)$-order $<_A$ where $F:=\freq[\mathcal A]^{t_{k,d}}_k$, this is by the choice of $t_{k,d}$. Now, as $f(k)\geq \bar J_{k,d}$ as well, all the conditions of \cref{transfer lemma} are met and so there exists an $(k,F)$-order $<_B$ of $\mathcal B$ and a border-preserving bijection between $A\setminus J^A$ and $B\setminus J^B$. The existence of  these orders and the order-preserving bijection between  $A\setminus J^A$ and $B\setminus J^B$ is enough to invoke \cref{main technical theorem} and conclude $$(\mathcal A,<_A) \equiv ^k_{\cfo}(\mathcal B,<_B)\,.$$ 
\end{proof}

\end{toappendix}

%% file: order.tex
\begin{definition}[$(k,F)$-order]
\label{definition:order}

Let $F$ be a set of $k$-neighbourhood types and $C^i_F$ be the set of $i$-contexts that are realisable with $k$-neighbourhood types of $F$. 

Let $\mathcal A$ be a graph and 
\begin{equation}\label{equation: the prec order on the parts}
\begin{aligned}
X^A \prec\;&
(LU^A_1 \prec LN^A_1)\prec (LU^A_2 \prec LN^A_2)\prec \cdots \prec (LU^A_{k^2} \prec LN^A_{k^2}) \\
&\prec LN^A_{k^2+1}\prec LN^A_{k^2+2}\prec \cdots \prec LN^A_{2k^2} \\
&\prec J^A \\
&\prec RN^A_{2k^2}\prec RN^A_{2k^2-1}\prec \cdots \prec RN^A_{k^2+1} \\
&\prec (RN^A_{k^2} \prec RU^A_{k^2})\prec (RN^A_{k^2-1} \prec RU^A_{k^2-1})\prec \cdots \prec (RN^A_1 \prec RU^A_1)
\end{aligned}
\end{equation}
be a partition of $A$ into $6k^2+2$ parts. We use the following terminology for the parts of this partition:
\begin{itemize}
	\item $X$ is called the \emph{extremity};
	\item $LN^A_i$, for $1\leq i\leq 2k^2$, are called \emph{left neighbouring segments};
	\item $LU^A_i$, for $1\leq i\leq k^2$, are called \emph{left universal segments};
	\item $RN^A_i$, for $1\leq i\leq 2k^2$, are called \emph{right neighbouring segments};
	\item $RU^A_i$, for $1\leq i\leq k^2$, are called \emph{right universal segments}; and
	\item $J^A$ is called the \emph{jungle}.
\end{itemize}
Let $\prec$ be the order on the parts of this partition, as depicted in \Cref{equation: the prec order on the parts}. We call this partition a \emph{$k$-partition} of $\mathcal A$.
Let $P = (\mathcal{B}_1, \mathcal{B}_2, \dots, \mathcal{B}_{6k^2+2})$ be the sequence of segments 
of a $k$-partition of $\mathcal{A}$ arranged according to the order $\prec$, \ie, $\mathcal{B}_1:= X, \mathcal{B}_2:= LU^A_1, \dots, \mathcal{B}_{6k^2+2}:= RU^A_1$.
For any element $x \in A$, let $\text{pos}_P(x)$ denote the index $j$ 
such that $x \in \mathcal{B}_j$. The \emph{segment distance} between two elements 
$x, y \in A$, denoted by $\bdist_P(x, y)$, is defined as:
\[
\bdist_P(x, y) = |\text{pos}_P(x) - \text{pos}_P(y)|
\]

For example, if $x$ and $y$ belong to different segments, $\bdist_P(x, y) = n + 1$, 
where $n$ is the number of segments strictly between the segment containing $x$ 
and the segment containing $y$. If $x$ and $y$ belong to the same segment, 
$\bdist_P(x, y) = 0$. For an integer $b\geq 0$, the $b$-segment neighbourhood of an element $x$ is the set of elements at segment distance at most $b$ from $x$ in $P$.

For $1\leq i\leq k^2$, by $L^A_i$ we mean $LN^A_i\cup LU^A_i$ and by $R^A_i$ we mean $RU^A_i\cup RN^A_i$. For a $k$-partition of $\mathcal A$ and $0\leq r\leq k$, we define \emph{safety parts} as
  $$S^{k-r}_A:= X^A \cup\bigcup_{i=1}^{(k-r)k} (L^A_i\cup R^A_i).$$

We say that an order $<_A$ is a
$(k,F)$-order of $\mathcal A$ if it is a $k$-partition of $\mathcal A$, meaning that for elements in different segments we use $\prec$ to compare them, or equivalently if $\text{pos}_P(x)< \text{pos}_P(y)$ then $x<_A y$.  Additionally, it satisfies the following properties:
\begin{description}
	\item[Universality: ] for every $1\leq r\leq k$, in each universal segment included in $S^{k-r+1}_A\setminus S^{k-r}_A$ there is an occurrence of every $(k-r)$-context $c\in C^{k-r}_F$.
	\item[Extremality: ] all elements whose $k$-neighbourhood type is not in $F$ are in $X_A$
        \item[Contraction:] the $k$-neighbourhood of an element $x$ is included in its $2k$-segment neighbourhood.
        \item[Tameness:] for every $1\leq r\leq k$, and every element $x\notin S^{k-r}_A$, we have $\mathds{C}^{k-r}_{(\mathcal A,<_A)}(x)\in C^{k-r}_F$.          
\end{description}

\end{definition}

Depending on $F$ and $k$, not all graphs admit a $(k,F)$-ordering. In the following lemma, we characterise the values of $k$ and the sets $F$ for which such an ordering exists. 
\begin{lemma}
\label{good graphs admit good orders}
Let $k,d$ be integers. 
There exists an integer $t_{k,d}$ such that every graph $\mathcal A$ with degree at most $d$ admits a $(k,F)$-order for $F:=\freq[\mathcal A]^{\geq t_{k,d}}_k$. Furthermove, for this order, the size of $|A\setminus J^A|$ can be bounded independently of $\mathcal A$ by a constant $\bar J_{k,d}$.
\end{lemma}

\begin{proof}
 Let $m_k:=2\cdot k^2 \cdot \bc^{k}_{d}\cdot\nc^{k}_{d}$ and
$r_k:=4k^2+2(k-1)$, and let $t_{k,d}$ be the threshold given by
\cref{lemma: lemma 5.6} for $d,k$, $m:=m_k$ and $r:=r_k$.
By the definition of $\nc^k_d$ and $\bc^k_d$, if a $k$-neighbourhood
type $\tau$ belongs to $F$, then $\mathcal A$ contains enough
occurrences of $\tau$
allowing us to realise all the possible $k$-contexts that appear in the
class of graphs with bounded degree $d$, each of them at least $2k^2$
times.
Moreover, these occurrences are pairwise at distance greater than
$r_k=4k^2+2(k-1)$.
As a consequence, the vertices with $k$-neighbourhood type $\tau$,
together with their whole $(k-1)$-neighbourhoods, are at distance
greater than $4k^2$.

More precisely, let $v_1$ and $v_2$ be two vertices with
$k$-neighbourhood type $\tau\in F$, chosen among these $m_k$
occurrences that are at distance $r_k$.
If $v_1^{+}$ and $v_2^{+}$ denote $v_1$ and $v_2$ together with their
respective $(k-1)$-neighbourhoods, then $v_1^{+}$ and $v_2^{+}$ are at
distance at least $4k^2$ in $\mathcal A$.

We construct a $(k,F)$-order of $\mathcal A$ by filling each part of a $k$-partition as following.
\begin{description}
	\item[$X^A$:] In $X^A$,  we put the set of all the vertices with $k$-neighbourhood types that do not belong to $F$. 
	
	We pick an arbitrary order on $X^A$.
	
      \item[$LU^A_i$ and $RU^A_{i}$:] For $0 \leq j\leq k-1$ and $jk+1\leq i\leq jk+k$,  in $LU^A_i$ and $RU^A_{i}$ we put a copy of all the contexts of $C^j_F$, i.e., the first $k$ universal segments contain all the $0$-contexts realisable with $F$, the second $k$ universal segments contain all the $1$-contexts realisable with $F$, and so on. 
      
      More precisely, let $l^c_{j,i}$ (resp. $r^c_{j,i}$), be a copy of the realised context $c\in C^j_F$ in $\mathcal A$, where all the elements of the set of elements of $l^c_{j,i}$ are amongst those particular $m_k$ copies of their $k$-neighbourhood types in $F$. We put $l^c_{j,i}$ (resp. $r^c_{j,i}$) in $LU^A_{i}$ (resp. $RU^A_{i}$).
      
      As we have chosen the elements of the set of elements of these contexts from those particular $m_k$ copies of them, we know that the distance between the elements of the set of elements of these contexts is greater than $4k^2$. In addition to that, by the definition of $t_{k,d}$, these occurrences of those neighbourhood types are also $4k^2+2(k-1)$ far apart from the \emph{rare} neighbourhood types, i.e., the elements with the neighbourhood types that are not in $F$. Hence, the set of elements of these contexts are also $4k^2+2(k-1)$ far from the elements of $X$.

       		Let $<_{C^j_F}$ be an arbitrary order on the $j$-contexts in $C^j_F$. That is if there are $c_1,\cdots,c_l\in C^j_F$, then we have $c_1<_{C^j_F}\cdots<_{C^j_F}c_l$.
       		       		
       		Accordingly, we define $<_A$ in $LU^A_i$ (resp. $RU^A_i$) to be $l^{c_1}_{j,i}<_Al^{c_2}_{j,i}<_A\cdots <_Al^{c_l}_{j,i}$ (resp. $r^{c_1}_{j,i}<_Ar^{c_2}_{j,i}<_A\cdots <_Ar^{c_l}_{j,i}$) which means that all the elements of the set of elements of $l^{c_1}_{j,i}$ are before all the elements of the set of elements of $l^{c_2}_{j,i}$ and so on up to $l^{c_l}_{j,i}$, and the same holds in $RU^A_{i}$.
	      \item[$LN^A_i$:] The construction of $LN^A_i$ for $1\leq i\leq 2{k^2}$,  is iterative and is done after all the universal segments are filled. In each $LN^A_i$ for $1\leq i\leq k^2$, we put the neighbours of its previous universal and  neighbouring segment, i.e., $LN^A_1:= N(X\cup LU^A_1)$ and $LN^A_i:=N(LN^A_{i-1}\cup LU^A_i)$ for $2\leq i\leq k^2$. In $LN^A_i$ for $k^2+1\leq i\leq 2k^2$, we put the neighbours of its previous neighbouring segment, i.e., $LN^A_i:=N(LN^A_{i-1})$.
                Note that as we required that the centre of all the neighbourhoods used in the construction of the contexts we put in the universal segments to be $r_k=4k^2+2(k-1)$ far-apart, implying the distance between any pair of elements from different neighbourhoods (and thus, in the set of elements of different contexts) is greater than $4k^2$. Therefore, we can safely fill the neighbouring segments following a universal segment, whose number is at most $2k^2$.

	 On all the $LN^A_i$, for $1\leq i\leq 2k^2$, we pick an arbitrary order $<_{LN^A_i}$.
	\item[$RN^A_i$:] We fill these segments exactly as we did for $LN^A_i$.
	On $RN^A_i$, for $1\leq i\leq 2k^2$, we pick an arbitrary order $<_{RN^A_i}$.

	\item[$J^A$:] In $J^A$, we put the rest of the elements of the structure, after we have filled all the other $6k^2+2$ parts. On $J^A$ we pick an arbitrary order $<_{J^A}$. 
\end{description}

By the construction, the \textbf{Universality} and \textbf{Extremality} conditions are satisfied.

Let $c$ be an $i$-context realised in $(\mathcal A,<_A)$ by a vertex $v$. The $i$-neighbourhood of $v$ is included in the $2i$-segment neighbourhood of $v$. Elements in an interval of the $i$-neighbourhood of $v$ are also in the $i$-context of $v$, hence, in $2i+2(i-1)$-segment neighbourhood of $v$, we may have an element of the $i-1$-neighbourhood of an element in an interval of $v$. Therefore, in the worst case scenario, an element of the set of elements of the $i$-context of $v$ can be as far as $2(\frac{i(i+1)}{2})$-segment distance from $v$. We conclude that all the elements of an $i$-context are in $i(i+1)$-segment neighbourhood from the realising vertex.

To show that \textbf{Tameness} condition is satisfied, note that $(k-r)$-context of an elements $x$ is included in $(k-r)^2+(k-r)$-segment neighbourhood of $x$, and for any element not in $S^{k-r}_A$, its segment distance from $X$ is at least $k(k-r)$, i.e., there are at least $k(k-r)$ segments between the segment of an element not in $S^{k-r}_A$ and $X$, which is always greater than or equal to $(k-r)^2+(k-r)$, for $r\leq k$. Therefore, the $(k-r)$-context of $x$ does not intersect $X$ and hence it only contains elements from $F$. To put it differently, the $(k-r)$-context of $x$ is realised using element of $k$-neighbourhood types in $F$, i.e., $\mathds{C}^{k-r}_{(\mathcal A,<_A)}(x)\in C^{k-r}_F$. The \textbf{Contraction} condition is satisfied by the construction and the way we fill neighbouring segments.

Looking at the material added in the various segments, a bound in terms of $k$ and $d$ on the size of all segments can be given, except for $J_A$. Crucially, this relies on the fact that the degree is bounded. This remark ensures the existence of a bound $\bar J_{k,d}$ for $|A\setminus J^A|$.
\end{proof}

Let $\mathcal A$ and $\mathcal B$ be two ordered graphs, and for $\varsigma\in \{A,B\}$, let
\begin{equation}
\begin{aligned}
X^\varsigma \prec\;&
(LU^\varsigma_1 \prec LN^\varsigma_1)\prec (LU^\varsigma_2 \prec LN^\varsigma_2)\prec \cdots \prec (LU^\varsigma_{k^2} \prec LN^\varsigma_{k^2}) \\
&\prec LN^\varsigma_{k^2+1}\prec LN^\varsigma_{k^2+2}\prec \cdots \prec LN^\varsigma_{2k^2} \\
&\prec J^\varsigma \\
&\prec RN^\varsigma_{2k^2}\prec RN^\varsigma_{2k^2-1}\prec \cdots \prec RN^\varsigma_{k^2+1} \\
&\prec (RN^\varsigma_{k^2} \prec RU^\varsigma_{k^2})\prec (RN^\varsigma_{k^2-1} \prec RU^\varsigma_{k^2-1})\prec \cdots \prec (RN^\varsigma_1 \prec RU^\varsigma_1).
\end{aligned}
\end{equation}
be a $k$-partition of $\varsigma$. Consider two orders $<_A$, $<_B$ on $\mathcal A$ and $\mathcal B$ that are compatible with their respective $k$-partition, meaning that if the part of element $a$ is $\prec$-smaller than the part of element $b$, then $a<_\varsigma b$.

A \emph{border-preserving bijection} $\phi$ between $(\mathcal A, <_A)$ and $(\mathcal B,<_B)$ is an isomorphism between $(\mathcal A,<_A)_{|A\setminus J^A}$ and $(\mathcal B,<_B)_{|B\setminus J^B}$ such that $x$ and $\phi(x)$ are in corresponding parts, which respects the partition ($x\in X^A$ iff $x\in X^B$, etc.).

\begin{lemma}
\label{transfer lemma}
Let $k$ be an integer and $\mathcal A$ and $\mathcal B$ be two graphs of bounded degree $d$ and $<_A$ be a $(k,F)$-order on $\mathcal A$, where $F:=\freq[\mathcal A]^{\geq t_{k,d}}_k$ ($t_{k,d}$ is given by Lemma~\ref{good graphs admit good orders}) such that
\begin{itemize}
\item $\mathcal A\equiv^{\bar J_{k,d}+1}_{\FO}\mathcal B$ (where $\bar J_{k,d}$ comes from Lemma~\ref{good graphs admit good orders})
\item $\mathcal A$ and $\mathcal B$ are $(k,t_{k,d})$-threshold equivalent.
\end{itemize} 
Then there exist a $(k,F)$-order $<_B$ of $\mathcal B$ and a border-preserving bijection between $(\mathcal A, <_A)$ and $(\mathcal B,<_B)$.
\end{lemma}

\begin{proof}
  
  The hypothesis ensures Duplicator can win the $(|A\setminus J^A|+1)$-round EF game between $\mathcal A$ and $\mathcal B$. In particular, if Spoiler plays only on $\mathcal A$ and covers $A\setminus J^A$ with pebbles, Duplicator's answer yields an isomorphism $\phi$ from $\mathcal A_{|A\setminus J^A}$ onto $\mathcal B_{S}$, where $S$ is the set of elements on which a pebble is played. We will need to refer to this configuration of the game later in the proof -- call it \texttt{Conf}. Note that Duplicator can win for one more round starting from \texttt{Conf}.

  We now transfer $<_A$ to $\mathcal B$ alongside $\phi$: the $k$-partition of $\mathcal B$ is given as $P^B:=\phi(P^A)$ for every part $P\neq J$, and $J^B:=B\setminus S$. As $<_B$, we choose any order that is compatible with this $k$-partition and for which $(\mathcal B,<_B)_{|B\setminus J^B}$ is isomorphic to $(\mathcal A,<_A)_{|A\setminus J^A}$. In other words, the images of $\phi$ are ordered in the same way as their preimages, and $<_B$ is completed arbitrarily on $J^B$.

  By construction, $\phi$ is a border-preserving bijection between $(\mathcal A,<_A)$ and $(\mathcal B,<_B)$.
  
  It only remains to prove that $<_B$ is a $(k,F)$-order on $\mathcal B$. Let us start by noting that the hypothesis ensures $F=\freq[\mathcal A]^{\geq t_{k,d}}_k = \freq[\mathcal B]^{\geq t_{k,d}}_k$.
  \begin{description}
  \item[Contraction:] since $\phi$ is a border-preserving bijection between $(\mathcal A,<_A)$ and $(\mathcal B,<_B)$, one only needs to show every neighbour of an element of $J^B$ is in $J^B\cup LN_{2k^2}^B \cup RN_{2k^2}^B$. Assume towards a contradiction that there exists an edge in $\mathcal B$ between some elements $a\in J^B$ and $b\in B\setminus(J^B\cup LN_{2k^2}^B \cup RN_{2k^2}^B)$. Remark that in \texttt{Conf}, there is a pebble placed on $b$. If Spoiler plays on $a$, Duplicator (who should win for one more round starting from \texttt{Conf}) would need to answer on an element of $J^A$ that is adjacent to the pebble corresponding to $b$ in $\mathcal A$. This is absurd, as this contradicts the contraction property of $<_A$.
  \item[Universality: ] let $1\leq r\leq k$. By universality of $<_A$, every universal segment included in $S^{k-r+1}_A\setminus S^{k-r}_A$ contains an occurrence $a_c$ of every $(k-r)$-context $c\in C^{k-r}_F$. To find an occurrence of $c$ in $S^{k-r+1}_B\setminus S^{k-r}_B$, we just need to consider the element $b_c$ corresponding to $a_c$ in \texttt{Conf}. Indeed, the contraction property of $<_B$ ensures that the context of $b_c$ is included in $B\setminus J^B$, that is, in the part of the structure where $\phi$ is an isomorphism.
  \item[Extremality and Tameness: ] every occurrence of a $k$-neighbourhood types $\tau\notin F$ are in part $X^B$. Indeed,
    \begin{itemize}
    \item $\mathcal A$ and $\mathcal B$ have the same number of occurrences of $\tau$,
    \item in $\mathcal A$, all these occurrences are in $X^A$, by tameness of $<_A$, and
    \item by contraction, an element of $X^A$ and its image by $\phi$ have the same $k$-neighbourhood type.
    \end{itemize}
    Let $1\leq r\leq k$, and let $b$ be an element outside of $S^{k-r}_B$. The $(k-r)$-context of $b$ in $(\mathcal B,<_B)$ is fully included in parts at segment distance at most ${(k-r)^2+(k-r)}$ from that of $b$. In particular, it does not intersect $X^B$, which means that all elements in it have a $k$-neighbourhood type in $F$, i.e. $\mathds{C}^{k-r}_{(\mathcal B,<_B)}(b)\in C^{k-r}_F$.
  \end{description}

\end{proof}

%% file: ModelChecking.tex
\begin{toappendix}
\section{Model Checking}
\label{section: model checking}

Let us now turn to the \emph{model checking} problem for \oicfo on a class $\mathcal C$ of graphs of bounded degree, which asks, on input consisting of a graph $\mathcal A$ and a sentence $\varphi\in\ \oicfo$, whether $\mathcal A\models\varphi$.

A well-known result by Seese~\cite{DBLP:journals/mscs/Seese96} states that on a class of bounded degree, \FO model checking is \emph{fixed-parameter tractable} (FPT for short), meaning that it can be solved in time $f(|\varphi|)\cdot O(|A|^c)$ for some computable function $f$ and some constant $c$ (in the aforementionned case, it even gets down $c=1$).

We have shown that when the degree is bounded, any sentence of \oicfo is equivalent to some sentence in \FO, but this does not immediately implies that the model checking problem for \oicfo is FPT on classes of bonuded degree, since the existence of a computable translation from \oicfo to \FO is not clear.
However, with a bit more work we are able to prove that this problem is indeed FPT:

\modelcheckingtheorem*

Let us first make an easy observation:

\begin{lemma}
  \label{lem:variable_namespace}
  Let $\varphi$ be a sentence of \oicfo of quantifier rank at most $k$. Then any variable $x_w^i$ appearing in $\varphi$ satisfies $i\leq k-|w|-1$.
\end{lemma}

\begin{proof}
  In the following, we define the depth of a node in a tree (or forest) as the distance to the root, and the depth of a tree (or forest) as the maximal depth among its nodes.

  Let us consider the labelled forest $F_\varphi$ that has a node labelled $x$ for each occurrence of a quantification on variable $x$ in $\varphi$, and where a node $u$ is the child of $v$ if the quantification $u$ is in the subformula rooted in $v$, with no quantifier in between.

  By hypothesis, $F_\varphi$ has depth at most $k-1$.
  
  In $F_\varphi$, a node labelled $x_w^0$ must have depth at least $|w|$, since labels $x_{w'}^0$, for all prefixes $w'$ of $w$, must appear in a branch from $x_w^0$ to its root.

  For similar reasons, the depth of a node labelled $x_w^i$ is at least $i$ more than the depth of $x_w^0$.

  All in all, a node labelled $x_w^i$ must have depth at least $|w|+i$. Combined with the fact that the depth of $F_\varphi$ is bounded by $k-1$, we get $|w|+i\leq k-1$.
\end{proof}

In the following, when we refer to the intervals of an outer context, we mean the intervals of the neighborhood underlying this outer context.

Let us prove the following lemma by induction on $n<k$:

\begin{lemma}
  \label{lem:mc_induction}
  Let $n<k$.
  There exists a procedure $\textsc{mc}_n$ that takes as input
  \begin{itemize}
  \item a finite non-empty set $S\subseteq\Sigma^\star\times\N$,
  \item a formula $\varphi\in\CFO[n]$ with free variables in $X:=(x_w^i)_{(w,i)\in S}$,
  \item a family $(o_w)_{(w,0)\in S}$, where $o_w$ is an outer $(k-|w|-1)$-context,
  \item a family $\big(\mathcal N^+_w\big)_{(w,0)\in S}$ of structures and
  \item a family of functions $(h_w)_{(w,0)\in S}$,
  \end{itemize}
  answers \textsc{yes} or \textsc{no}, and has the following property.

  Let $S\subseteq\Sigma^\star\times\N$, $(\mathcal A,<)$ be an ordered graph and let $\nu:x_w^i\mapsto a_w^i$ be a valuation from $X$ to $A$ such that for every $(w,i)\in V,a_w^i\in\boule[k-|w|-1][a_w^0]$. Let, for each $(w,0)\in S$,
  \begin{itemize}
  \item $o_w:=\ocontext[k-|w|-1][a_w^0]$,
  \item $\mathcal N^+(a_w^0)$ be the structure $(\mathcal A,<)|_{\boule[k-|w|-1][a_w^0]}$ enriched with a constant symbol $c_w^i$ interpreted as $a_w^i$ for each $i\in\N$ such that $(w,i)\in S$ ($\boule[k-|w|-1][a_w^0]$ must contain all such $a_w^i$ for a valid valuation $\nu$, according to Lemma~\ref{lem:variable_namespace}), and
  \item $h_w$ be the function associating
    \begin{itemize}
    \item with each interval $\mathcal I\in\OInt[\boule[k-|w|-1][a_w^0]]$ the set of $\alpha\in\Sigma$ such that $(w\alpha,0)\in S$ and $a_{w\alpha}^0\in\mathcal I$, and
    \item with each element $a$ of $\neigh[k-|w|-1][a_w^0]$ the set of $\alpha\in\Sigma$ such that $(w\alpha,0)\in S$ and $a_{w\alpha}^0=a$.
    \end{itemize}
  \end{itemize}

  Then $\textsc{mc}_n$, on input $S$, $\varphi$, $(o_w)_{(w,0)\in S}$, $\big(\mathcal N^+_w\big)_{(w,0)\in S}$ and $(h_w)_{(w,0)\in S}$, outputs \textsc{yes} if and only if $(\mathcal A,<,\nu)\models\varphi$.
\end{lemma}

\begin{proof}
  $\textsc{MC}_0$ receives as input a quantifier-free formula $\varphi$. Each atomic formula is of the form
  \begin{itemize}
  \item $E(x_w^i,x_w^j)$ or $x_w^i\sim x_w^j$ for $(w,i),(w,j)\in S$ and ${\sim}\in\{<,=\}$, which can be evaluated in $(\mathcal N^+_w)$, or
  \item $x_{w\alpha}^0\sim x_w^i$ for $(w\alpha,0),(w,i)\in S$ and ${\sim}\in\{<,>,=\}$, which can be evaluated according to $h_w$.
  \end{itemize}
  $\textsc{MC}_0$ then computes and returns the boolean combination of these evaluations, which indeed corresponds to whether $(\mathcal A,<,\nu)\models\varphi$.

  Let us turn to $\textsc{MC}_{n+1}$, assuming the existence of a suitable procedure $\textsc{MC}_n$. The input formula $\varphi\in\CFO[n+1]$, with free variables indexed in $S$, is a boolean combination of formulas of the form
  \begin{enumerate}[(a).]
  \item\label{enum:intro} $\exists x_{w\alpha}^0 \psi$, where $\alpha\in\Sigma$ and $\psi\in\CFO[n]$ has free variables indexed in $S\cup\{(w\alpha,0)\}$, and 
  \item\label{enum:cont} $\exists x_w^i, E(x_w^i,x_w^j)\land \theta$, where $(w,j)\in S$.
  \end{enumerate}
  Once again, we only need to show that $\textsc{mc}_{n+1}$ answers correctly on formulas of type (\ref{enum:intro}) and (\ref{enum:cont}), the end result being computed by evaluating a boolean combination of the results of these subprocedures.

  \begin{description}
  \item[Case (\ref{enum:intro}):] assume first $\varphi=\exists x_{w\alpha}^0 \psi$. Remark, given an interval $\mathcal I\in\OInt[\context[k-|w|-1][a_w^0]]$ and two elements $a_{w\alpha}^0, b_{w\alpha}^0\in\mathcal I$ such that $\ocontext[k-|w|-2][a_{w\alpha}^0]=\ocontext[k-|w|-2][b_{w\alpha}^0]$, that the induction hypothesis guarantees \[(\mathcal A,<,\nu+x_{w\alpha}^0\mapsto {a_{w\alpha}^0})\models\psi\text{ iff }(\mathcal A,<,\nu+x_{w\alpha}^0\mapsto b_{w\alpha}^0)\models\psi\,,\] since the arguments given to $\textsc{mc}_n$ are the same in both cases, and thus its answer is shared. In other words, to know whether an element in a given interval is a good candidate to witness $\exists x_{w\alpha}^0 \psi$, one only needs looking at its outer $(k-|w|-2)$-context.

    With that in mind, $\textsc{mc}_n$ does the following:
    \begin{itemize}
    \item for each interval $\mathcal I$ in $o_w$ and each outer $(k-|w|-2)$-context $o_{w\alpha}$ which is realised in $\mathcal I$ (as specified by the data contained in $o_w$), it makes a call to $\textsc{MC}_n$ with arguments
      \begin{itemize}
      \item $S\cup\{(w\alpha,0)\}$
      \item $\psi$
      \item $(o_{w'})_{(w',0)\in S} + o_{w\alpha}$
      \item $\big(\mathcal N^+_{w'}\big)_{(w',0)\in S} + \mathcal N^+_{w\alpha}$, where $\mathcal N^+_{w\alpha}$ is extracted from $o_{w\alpha}$ and contains the $(k-|w|-2)$-neighborhood of the center of $o_{w\alpha}$, with constant $c_{w\alpha}^0$ interpreted as its center
      \item $(h_{w'})_{(w',0)\in S, w'\neq w} + h_w:
        \begin{cases}
          \mathcal I&\mapsto h_w(\mathcal I)\cup\{\alpha\}\\
          x\neq\mathcal I&\mapsto h_w(x)
        \end{cases}
          $
      \end{itemize}
      It follows from the previous remark that this call to $\textsc{MC}_n$ answers \textsc{yes} exactly when some (equivalently, every) occurrence $a_{w\alpha}^0$ in $\mathcal I$ with outer context $o_{w\alpha}$ is such that $(\mathcal A,<,\nu+x_{w\alpha}^0\mapsto a_{w\alpha}^0)\models\psi$. 

    \item for each element $e$ of the $(k-|w|-1)$-neighborhood of the center of $o_w$ with outer $(k-|w|-2)$-context $o_{w\alpha}$ (as specified by the data contained in $o_w$), it makes a call to $\textsc{MC}_n$ with arguments
      \begin{itemize}
      \item $S\cup\{(w\alpha,0)\}$
      \item $\psi$
      \item $(o_{w'})_{(w',0)\in S} + o_{w\alpha}$
      \item $\big(\mathcal N^+_{w'}\big)_{(w',0)\in S} + \mathcal N^+_{w\alpha}$, where $\mathcal N^+_{w\alpha}$ is extracted from $o_{w\alpha}$ and contains the $(k-|w|-2)$-neighborhood of the center of $o_{w\alpha}$, with constant $c_{w\alpha}^0$ interpreted as its center
      \item $(h_{w'})_{(w',0)\in S, w'\neq w} + h_w:
        \begin{cases}
          e&\mapsto h_w(e)\cup\{\alpha\}\\
          x\neq e&\mapsto h_w(x)
        \end{cases}
          $
      \end{itemize}
    \end{itemize}
    $\textsc{MC}_{n+1}$ then answers \textsc{yes} if and only if at least one of these calls to $\textsc{mc}_n$ returns \textsc{yes}.

    Since any witness to $\exists x_{w\alpha}^0 \psi$ must either fall within an interval of $\OInt[\ocontext[k-|w|-1][a_w^0]]$ or in \boule[k-|w|-1][a_w^0], $\textsc{mc}_{n+1}$ answers \textsc{yes} if and only if $(\mathcal A, <,\nu)\models\exists x_{w\alpha}^0 \psi$.

  \item[Case (\ref{enum:cont}):] it now remains to study the case where $\varphi=\exists x_w^i, E(x_w^i,x_w^j)\land \theta$ for some $(w,j)\in S$.

    In that case, for every element $e$ from $\mathcal N^+_w$ that is adjacent to the constant $c_w^j$, $\textsc{mc}_{n+1}$ does a call to $\textsc{mc}_n$ with input
    \begin{itemize}
    \item $S\cup\{(w,i)\}$
    \item $\theta$
    \item $(o_{w'})_{(w',0)\in S}$
    \item $\big(\mathcal N^+_{w'}\big)_{(w',0)\in S,w'\neq w} + \mathcal N^\star_{w}$, where $\mathcal N^\star_{w}$ is obtained from the structure $\mathcal N^+_w$ by adding the constant $c_{w}^i$, interpreted as $e$
    \item $(h_{w'})_{(w',0)\in S}$
    \end{itemize}
    $\textsc{MC}_{n+1}$ then answers \textsc{yes} if and only if at least one of these calls to $\textsc{mc}_n$ returns \textsc{yes}.

    Since the neighbors of $a_w^j$ are the only potential witnesses to $\exists x_w^i, E(x_w^i,x_w^j)\land \theta$, $\textsc{mc}_{n+1}$ answers \textsc{yes} if and only if $(\mathcal A, <,\nu)\models\exists x_w^i, E(x_w^i,x_w^j)\land \theta$.
  \end{description}
\end{proof}

Before concluding the proof of Theorem~\ref{th:model-checking}, we need the following lemma:

\begin{lemma}
  \label{lem:outer_contexts}
  Given an ordered graph $(\struct,<)$ and an integer $k$, the set of outer $k$-contexts realised in $(\struct,<)$ can be computed in FPT time.
\end{lemma}

\begin{proof}
  Let us exhibit a recursive procedure to compute the function $v\mapsto\ocontext[n][v]$ in a time that is quadratic in $|A|$.

  If $n=0$, we just have to check the colour of $a$: one pass through \struct is enough, and the time taken is linear in $|A|$.

  For $n>0$, let us first make a recursive call in order to compute the function $a\mapsto\ocontext[n-1]$. Then to compute $\ocontext[n][v]$, we go through each interval of \OInt[\boule[n][v]], and collect the set of outer $(n-1)$-contexts realised in this interval. This allows us to compute function $f$. Computing $g$ is trivial given the result of the recursive call.

  Doing this computation for vertex $v$ requires time linear in $|A|$. All in all, the function $v\mapsto\ocontext[n][v]$ can be computed in time quadratic in $|A|$.

  To compute the set of outer $k$-contexts realised in $(\struct,<)$, one just needs to collect the images of $v\mapsto\ocontext[k][v]$.
\end{proof}

We are now ready to prove that model checking \oicfo is FPT on classes of bounded degree.

\begin{proof}[Proof of Theorem~\ref{th:model-checking}]
  Let us consider an ordered graph $\struct\in\mathcal C$, where $\mathcal C$ is a class of bounded degree, and a sentence $\varphi\in\oicfo$.

  Since one can decompose $\varphi$ as a boolean combination of formulas of the form $\exists x_\epsilon^0\psi$, where $\psi\in\CFO[k-1]$, it is enough to be able to model check formulas of this form. In what follows, we thus consider $\varphi=\exists x_\epsilon^0\psi$, where $\psi\in\CFO[k-1]$ has only $x_\epsilon^0$ as free variable.

  We start by endowing $\struct$ with an arbitrary order $<$, and computing the set $O$ of outer $(k-1)$-contexts realised in $(\struct,<)$. This can be done in FPT time according to Lemma~\ref{lem:outer_contexts}.
  
  For each outer $(k-1)$-context $o_\epsilon^0\in O$, we make a call to $\textsc{mc}_{k-1}$ with input
  \begin{itemize}
  \item $S=\{(\epsilon,0)\}$
  \item $\psi$
  \item $o_\epsilon^0$
  \item $\mathcal N^+_{(\epsilon,0)}$, which contains the $(k-1)$-neighborhood of the center of $o_\epsilon^0$ (which can be gathered from $o_\epsilon^0$), with constant $c_\epsilon^0$ interpreted as its center
  \item $h_\epsilon:x\mapsto\emptyset$
  \end{itemize}
  We then answer \textsc{yes} if and only if at least one of these calls answers \textsc{yes}.  According to Lemma~\ref{lem:mc_induction}, whether an element is a suitable witness to $\exists x_\epsilon^0\psi$ depends only on its outer $(k-1)$-context, which ensure the correction of the model checking procedure.
  
  For the running time, remark that the number of these calls does not depend on \struct, as it is bounded by the number of $(k-1)$-contexts. Similarly, the time taken by each individual call is independent on \struct, as the size of the input depends only on $\varphi$ and $\mathcal C$.

  We have thus exhibited a model checking procedure for $\oicfo[k]$ on $\mathcal C$ that runs in time quadratic in $|A|$.
\end{proof}

\end{toappendix}

%% file: OtherProperties.tex
\begin{toappendix}
\section{Separating example}
\label{section: further properties of cfo}

We have shown that \oicfo is included in \FO on classes of graphs of bounded degree.
This raises the question whether \oicfo is able to go beyond \FO when the degree is not assumed to be bounded.

We answer this question by the positive, by relying on the (unpublished, but folklore) property introduced by Gurevich to prove that $\FO\subsetneq\ \oifo$:

\seperatingexample*

\begin{proof}[Proof sketch]
To prove that \oifo is stricly more expressive than plain \FO, Gurevich exhibited a sentence of \oifo defining a property which is not \FO-definable. In the following, we refer to the presentation of this construction that can be found in \cite[Section 5.2]{DBLP:books/sp/Libkin04}.

The property \qeven separating \oifo from \FO is defined over vocabulary $\{\subseteq\}$, where $\subseteq$ is a binary relation symbol.\footnote{We refer in this proof to the binary relation symbol as $\subseteq$ instead of $E$ in the interest of uniformity of notation with \cite[Section 5.2]{DBLP:books/sp/Libkin04}. This obviously has no consequence on the expressive power of \oicfo.} A structure satisfies property \qeven iff it is isomorphic to a boolean algebra $(2^X,\subseteq)$ for some finite set $X$ of even size. 

An Ehrenfeucht-Fra\"{i}ssé argument shows that \qeven is not definable in \FO \cite[Lemma 5.5]{DBLP:books/sp/Libkin04}.

To prove the theorem, we need only show that \qeven is not only definable by some sentence $\phi$ of \oifo \cite[Lemma 5.4]{DBLP:books/sp/Libkin04}, but can already be defined in \oicfo. To that end, let us first consider a formula $\psi$ in $\cfo$ stating that all the elements are related by the relation $\subseteq$ with  $x_\epsilon^0$: such a formula can be written down by 
\begin{enumerate}[(i)]
\item first ensuring that all elements at distance at most $2$ from $x_\epsilon^0$ in the Gaifman graph are in fact related by the relation $\subseteq$ with $x_\epsilon^0$ (making of $x_\epsilon^0$ a central vertex of its connected component, which has diameter at most $2$),
\item\label{enu:cc_bounded_diam} then by stating, in a similar fashion, that every connected component has diameter at most $2$,
\item finally by making sure, using the order, that the connected component of $x_\epsilon^0$ is the only one. This can be enforced by saying that the smallest element of a connected component (a definable property in view of (\ref{enu:cc_bounded_diam})) cannot have a smaller element in the whole structure.
\end{enumerate}

With $\psi$ defined, property \qeven is then captured by the \oicfo sentence \[\Phi:=\exists x^0_\epsilon\ \psi\land\phi[\cdot\subseteq x^0_\epsilon]\,,\]
where $\phi[\cdot\subseteq x^0_\epsilon]$ is the formula obtained from $\phi$ by replacing each subformula \[\exists x\ \theta\quad\text{by}\quad\exists x_\epsilon^i\ x_\epsilon^i\subseteq x_\epsilon^0\land\theta\] and each \[\forall x\ \theta\quad\text{by}\quad\forall x_\epsilon^i\ x_\epsilon^i\subseteq x_\epsilon^0\to\theta\,,\] $i$ being the smallest integer for which $x_\epsilon^i$ is a new variable.
Indeed, sentence $\Phi$ belongs to $\cfo$ by construction; it then suffices to notice the equivalence of $\phi$ (which is itself in \oifo) and $\Phi$ to get that $\Phi$ is in \oicfo.
\end{proof}
\end{toappendix}